\documentclass[a4paper,11pt]{amsart}

\usepackage{hyperref}       
\usepackage{url}            
\usepackage{amsthm}
\usepackage{amsfonts}
\usepackage{amsmath}
\usepackage{amssymb}
\usepackage{color}
\usepackage{multirow}
\usepackage{float}
\usepackage[noadjust]{cite}
\usepackage{mathtools}
\usepackage{blkarray}
\usepackage{booktabs}
\usepackage{hhline}
\usepackage{enumitem}
\usepackage{float}
\usepackage{caption}
\usepackage{mathrsfs}
\usepackage{graphicx}
\usepackage{comment}
\usepackage{cleveref}

\usepackage{comment}

\usepackage[margin=2.7cm]{geometry}

\theoremstyle{definition}

\newtheorem{theorem}{Theorem}[section]
\newtheorem{corollary}[theorem]{Corollary}
\newtheorem{lemma}[theorem]{Lemma}
\newtheorem{definition}[theorem]{Definition}
\newtheorem{proposition}[theorem]{Proposition}
\newtheorem{remark}[theorem]{Remark}
\newtheorem{notation}[theorem]{Notation}

\newtheorem{example}[theorem]{Example}

\newtheorem{conjecture}[theorem]{Conjecture}

\usepackage{array}
\newcolumntype{P}[1]{>{\centering\arraybackslash}p{#1}}
\usepackage{multirow, boldline,makecell}
\usepackage{enumitem}
\usepackage{MnSymbol}
\usepackage[symbol]{footmisc}

\def\C{\mathcal{C}}

\def\A{\mathcal{A}}
\def\B{\mathcal{B}}

\def\L{\mathcal{L}}

\def\H{\textup{H}}

\def\supp{\textup{supp}}

\def\Fq{\mathbb{F}_q}

\def\F{\mathbb{F}}

\def\<{\left<}
\def\>{\right>}
\def\diag{\textup{diag}}
\def\rk{\textup{rk}}
\def\trk{\textup{trk}}

\def\wt{\textup{wt}}

\def\dH{\textup{d}^\textup{H}}

\def\trkdef{\Delta^\textup{trk}}
\def\Tr{\textup{Tr}}
\def\FqnR{\Fq^{n\times R}}

\def\FqmR{\Fq^{m\times R}}
\def\wt{\textup{wt}^\textup{H}}
\def\drk{\textup{d}^{\rk}}
\def\Hdef{\Delta^\textup{H}}

\usepackage[dvipsnames,x11names,table]{xcolor}

\renewcommand{\L}{\mathscr{L}}

\usepackage{tikz}
\usepackage{pgfplots}

\newcolumntype{C}[1]{>{\centering\arraybackslash$}p{#1}<{$}}
\newcolumntype{L}[1]{>{\arraybackslash$}p{#1}<{$}}

\title{Constructions of Rank-Metric Codes of Small Tensor Rank}
\author{Matteo Bonini}
\address{Department of Mathematical Sciences, Aalborg University}
\email{mabo@math.aau.dk}

\author{Eimear Byrne}
\address{School of Mathematics and Statistics, University College Dublin}
\email{ebyrne@ucd.ie}

\author{Giuseppe Cotardo}
\address{Department of Mathematics, Virginia Tech}
\email{gcotardo@vt.edu}

\date{}

\begin{document}
	
	\maketitle
	
	\begin{abstract}
		Rank-metric codes are subspaces of matrices over finite fields endowed with the rank metric and admit a natural tensorial representation. The tensor rank provides a measure of the minimal size of a decomposition of a code into rank-one tensors. Kruskal showed that the tensor rank of a rank-metric code of dimension $k$ and minimum rank distance $d$ is at least $k + d - 1$, and codes meeting this bound with equality are called minimal tensor rank (MTR) codes. It is known from algebraic complexity theory that the existence of an MTR code implies the existence of a maximum distance separable (MDS) code. In this work, we establish new results relating the tensor rank of a rank-metric code to the parameters of associated linear codes in the Hamming metric and introduce the notion of tensor rank defect. We then develop new constructions of rank-metric codes with small tensor rank defect using algebraic geometry (AG) codes.
	\end{abstract}

	\section{Introduction}
	
	Rank-metric codes are finite-dimensional subspaces of the space of $m \times n$ matrices over a finite field $\Fq$, where the ambient space is endowed with the rank metric. The rank distance between two matrices is defined to be the rank of their difference. Well-known classes of optimal rank-metric codes, such as the Gabidulin maximum rank distance codes and their generalizations, have very efficient decoders. Rank-metric codes have been extensively studied in connection with applications such as code-based cryptography, network coding, and distributed storage (see \cite{bartz+} and the references therein). Moreover, rank-metric codes are closely related to a variety of combinatorial structures, including Ferrers diagrams, subspace designs, and $q$-polymatroids (see, for example, \cite{byravsiam,gorla2019rank,NERI2024105937,cotardo2023diagonals}).\newline
	Every matrix code can be viewed as the slice space of a $3$-tensor, which allows one to study rank-metric codes via tensorial representations. In analogy with generator and parity-check matrices for linear codes in $\Fq^n$, one can consider generator and parity-check tensors for rank-metric codes. This perspective was developed in \cite{byrne2019tensor}. The tensor representation of a rank-metric code was also used in \cite{franch} to describe a family of generalized low-rank parity check codes with applications to code-based cryptography. 
	Any two generator tensors of a rank-metric code $\mathcal{C}$ have the same tensor rank. This allows one to define the tensor rank of $\mathcal{C}$, denoted $\trk(\mathcal{C})$, as the tensor rank of any of its generator tensors. Although computing the tensor rank of a $3$-tensor is NP-complete \cite{HASTAD1990644}, The tensor rank of a matrix code gives a measure of its storage and encoding complexity and is hence a relevant parameter \cite{byrne2019tensor,byrne2023tensor}. There have been several papers focused on determining or estimating tensor rank for algebras specific small parameters (see, for example, \cite{bartoli2022non,LAVRAUW202299,byrne2024constructions}).
	A fundamental lower bound on the tensor rank of a rank-metric code was established by Kruskal~\cite{kruskal1977three}, who showed that if $\mathcal{C}$ has dimension $k$ and minimum rank distance $d$, then~$\trk(\mathcal{C}) \geq k + d - 1$. Codes attaining this bound are called minimal tensor rank (MTR) codes. The tensor rank defect of $\mathcal{C}$, defined to be $\trkdef(\mathcal{C}) = \trk(\mathcal{C}) - (k + d - 1)$, gives a measure of its deviation from the MTR property. Note that $\trk(\C)=0$ only if a maximum distance separable $[k+d-1,k,d]$ code exists, so this problem has links to the MDS conjecture \cite{segre1955curve}.
	
	One approach to studying tensor rank relies on a connection between rank-metric codes and classical linear codes in the Hamming metric, which was already observed in \cite{brockett1978optimal}. Specifically, any rank-metric code $\mathcal{C}$ of tensor rank at most $R$ admits a representation of the form
	\[
	\mathcal{C} = \{V \diag(c) W^{\textup{t}} : c \in C\},
	\]
	where $V$ and $W$ are fixed matrices and $C$ is a linear code of length $R$. This representation induces an epimorphism from $C$ onto 
	$\mathcal{C}$. Moreover, the minimum rank distance of $\mathcal{C}$ is bounded above by the minimum Hamming distance of $C$. Consequently, constructing rank-metric codes with small tensor rank is closely related to constructing linear codes of small length with prescribed dimension and distance. 
	
	The main problem in this approach lies in estimating the minimum rank distance of $\mathcal{C}$, or obtaining suitable upper bounds on it. A key part of this is to control the dimension of intersections of the form
	\[
	C_V^\perp \cap (c * C_W),
	\]
	where $C_V$ and $C_W$ are linear codes generated by the rowspace of the matrices $V$ and $W$, respectively and $c*C_W=\{c*c' = (c_1c_1',\dots,c_Rc_R') : c' \in C_W\}$. Problems of this type naturally lead to questions about intersections of linear codes, which have been studied to enable constructions of codes with particular properties, such as complementary dual codes \cite{mesnager2017complementary} or $\ell$-DLIP codes \cite{bhowmick2024linear,bhowmick2025}. Evaluation codes and algebraic geometry codes have been studied in this context.
	
	In this paper, we build on the framework introduced in \cite{byrne2019tensor} and develop new constructions of rank-metric codes with specified tensor rank defect by exploiting properties of AG codes.
	\newline
	The paper is organised as follows. In Section~\ref{sec:pre}, we provide background material on linear codes, algebraic geometry (AG) codes, rank-metric codes, and the tensorial representation of matrix codes. In Section~\ref{sec:genten}, we introduce the notion of tensor rank defect and develop the general framework relating rank-metric codes to linear codes via tensor representations. In particular, we establish results connecting the tensor rank defect to the Hamming parameters of associated linear codes.
	We give a construction of MTR codes based on Cauchy codes, which extends \cite[Theorem~5.7]{byrne2019tensor} to a strictly larger class of codes. 
	In Section~\ref{sec:intersection}, we present new constructions of rank-metric codes using AG codes. By exploiting intersection properties of AG codes and controlling suitable divisors, we obtain families of codes with small tensor rank defect.

	\section{Preliminaries}\label{sec:pre}
	
	Throughout the paper, let $n, m, R$ be positive integers satisfying $2 \leq n , m \leq R$. Let $q$ be a fixed prime power, and let $\Fq$ denote the finite field with $q$ elements. We denote by $\Fq^R$ the space of vectors of length $R$ over $\Fq$, and by $\Fq^{n \times m}$ the space of $n \times m$ matrices with entries in $\Fq$. We recall some basic definitions and well-known facts on \textit{Hamming metric codes} and \textit{rank-metric codes}. The \textbf{Hamming support} of a vector $v \in \Fq^R$ is defined to be $\supp(v) = \{ i \in \{1, \ldots, R\} : v_i \neq 0 \}$, and its \textbf{Hamming weight} is $\wt(v) = |\supp(v)|$. Moreover, for a matrix $M \in \Fq^{n \times m}$, we denote by $\rk(M)$ its rank.
	
	\begin{definition}
		A \textbf{linear code} is a subspace $C \leq \Fq^R$. The \textbf{minimum (Hamming) distance} of a non-zero linear code $C$ is defined to be
		\[
		\dH(C) = \min\{ \wt(v) : v \in C,\ v \neq 0 \}.
		\]
		For $C = \{0\}$, we define $\dH(C) = R + 1$. 
		If $C$ has dimension $k$, then we refer to is as an $\Fq$-$[R,k]$ code (or $\Fq$-$[R,k]$ code, in short). If $d=\dH(C)$ we refer to $C$ as an $\Fq$-$[R,k,d]$ code.
		The \textbf{support} of $C$ is defined to be the union of the supports of its codewords, that is, $\supp(C) = \{ i \in \{1, \ldots, R\} : \exists\, c \in C \text{ such that } c_i \neq 0 \}$. The \textbf{dual} code of $C$ is $C^\perp = \{ u \in \Fq^R : u \cdot v = 0 \text{ for all } v \in C \}$,
		where $\cdot$ denotes the standard inner product.
	\end{definition}
	
	It is well known that the dual code $C^\perp$ of an $[R,k]$ code $C$ is an $\Fq$-$[R, R - k]$ code. We recall the Singleton bound for linear codes.
	
	\begin{theorem}[Singleton bound]
		If $C$ be an $\Fq$-$[R,k,d]$ code,
		then $k \leq R - d + 1$.
	\end{theorem}
	
	We say that $C$ is an \textbf{MDS} (\textbf{maximum distance separable}) code if it meets the Singleton bound with equality. A major open problem in coding theory is to answer the following conjecture.
	
	\begin{conjecture}[MDS Conjecture]\label{conj:MDSconj}
		Let $C$ be an $\Fq$-linear MDS code of dimension $k$. If~$k \leq q$, then the length $R$ of $C$ satisfies
		\[
		R \leq 
		\begin{cases}
			q + 2 & \text{if } q = 2^h \text{ and } k \in \{3, q - 1\},\\
			q + 1 & \text{otherwise}.
			
		\end{cases}
		\]
	\end{conjecture}

	The MDS Conjecture has been proven in many specific cases and parameter ranges. Table~\ref{tab:MDSConjectureValues} summarizes the known results for $k \geq 4$. For a more detailed and unified exposition of the results concerning the MDS Conjecture in these cases, we refer the reader to \cite{ball2020arcs}. Despite significant progress, the MDS Conjecture remains open in general for larger values of $k$ and~$q$.

	\begin{table}[h]
		\centering
		\renewcommand{\arraystretch}{1.5}
		\begin{tabular}{|c|c|c|}\hline
			$k$ & $q$ & Reference \\\hline
			
			4 
			& -- 
			& Segre (1955), Casse (1969) -- \cite{segre1955curve,casse1969solution} \\\hline
			
			5
			& -- 
			& Segre (1967), Casse (1969) -- \cite{casse1969solution,segre1967introduction} \\\hline
			
			6 
			& $2^h$
			& Casse--Glynn (1984) -- \cite{casse1982solution} \\\hline
			
			7 
			& $2^h,\ q\ge 16$
			& Kaneta--Maruta (1991) -- \cite{maruta1991uniqueness} \\\hline
			
			$\frac12\sqrt{q}+\frac{15}{4}$
			& $2^h$
			& Storme--Thas (1993) -- \cite{storme1993mds} \\\hline
			
			$\frac14\sqrt{pq} - \frac{29}{16}p + 4$
			& $q = p^{2h+1}$, $p$ odd
			& Voloch (1991) -- \cite{voloch1991complete} \\\hline
			
			$\sqrt{q} - \frac{\sqrt{q}}{p} + 2$
			& $q = p^{2h}$, $p$ odd
			& Ball--Lavrauw (2018) -- \cite{ball2018planar} \\\hline
			
			$\frac12\sqrt{q}$
			& $q = p^h,\ p\ge 5$ odd
			& Hirschfeld--Korchmáros (1996) -- \cite{hirschfeld1996embedding} \\\hline
			
			$\frac12\sqrt{q} + 2$
			& $q\ge 529$, $q$ odd, $q\neq 3^6, 5^5$
			& Hirschfeld--Korchmáros (1998) -- \cite{hirschfeld1998number} \\\hline
			
			$q$
			& $q=p$ prime
			& Ball (2012) -- \cite{ball2012sets} \\\hline
			
			$2p - 2$
			& $q=p^h,\ h>1$
			& Ball--De Beule (2012) -- \cite{ball2017subsets} \\\hline
			
		\end{tabular}\vspace{0.5em}
		\caption{Summary of all known parameter ranges $(k,q)$ for which the MDS Conjecture has been established.}
		\label{tab:MDSConjectureValues}
	\end{table}

	We recall the following well-known notion.
	
	\begin{definition}[Hamming defect]
		The \textbf{Hamming defect} of an $[R,k,d]$ code $C$ is defined to be
		$\Hdef(C) = R - k - d + 1.$
	\end{definition}
	
	In particular, a code is MDS if and only if its Hamming defect is equal to $0$. An $[R,k,d]$ code $C$ is said to be \textbf{almost MDS} (AMDS) if it has minimum distance $d = R - k$, that is, if~$\Hdef(C)=1$. The code $C$ is said to be \textbf{near MDS} (NMDS) if both $C$ and its dual code are AMDS.
	
	\subsection{Algebraic Geometry Codes}

	We now recall some basic definitions and results concerning algebraic geometry (AG) codes. A standard reference for this theory is \cite{stichtenoth2009algebraic}.
	
	Throughout the paper, we let $\mathcal{X}$ be a projective, nonsingular, geometrically irreducible algebraic curve of genus $g$ defined over the finite field $\F_q$, and let $\F_q(\mathcal{X})$ denote its function field, that is, the field of rational functions on $\mathcal{X}$. The $\F_q$-rational points of $\mathcal{X}$ correspond to the~$\F_q$-rational places of $\F_q(\mathcal{X})$, and we denote this set by $\mathcal{X}(\F_q)$. A \textbf{divisor} $D$ on $\mathcal{X}$ is a formal sum of the form
	\[
	D = \sum_{P \in \mathcal{X}} n_P P, \quad n_P \in \mathbb{Z},
	\]
	for which only finitely many coefficients $n_P$ are nonzero. The \textbf{degree} of $D$ is defined to be
	\[
	\deg(D) = \sum_{P \in \mathcal{X}} n_P.
	\]
	A divisor is said to be \textbf{effective}, denoted $D \ge 0$, if $n_P \ge 0$ for all $P \in \mathcal{X}$. For a nonzero function $f \in \F_q(\mathcal{X})$, the \textbf{principal divisor} associated with $f$ is defined to be
	\[
	(f) = \sum_{P \in \mathcal{X}} v_P(f)\, P,
	\]
	where $v_P(f)$ denotes the valuation of $f$ at $P$, that is, the order of vanishing (or the negative of the order of the pole) of $f$ at $P$. In particular, every principal divisor has degree zero. A divisor is said to be \textbf{$\F_q$-rational} if it is invariant under the action of the Frobenius automorphism. Equivalently, $D$ is $\F_q$-rational if it is fixed under this action.
	
	Given a divisor $D = \sum_{i=1}^R n_i P_i$ with $P_i \in \mathcal{X}(\F_q)$ and $n_i \in \mathbb{Z}$, the \textbf{support} of $D$ is defined to be
	\[
	\supp(D) = \{ P_i \in \mathcal{X}(\F_q) : n_i \neq 0 \}.
	\]
	
	Given two divisors $D = \sum_{P \in \mathcal{X}} n_P P$ and $D' = \sum_{P \in \mathcal{X}} m_P P$, their \textbf{sum} and \textbf{difference} are defined componentwise by
	\[
	D \pm D' = \sum_{P \in \mathcal{X}} (n_P \pm m_P)\, P.
	\]
	The \textbf{greatest common divisor} and \textbf{least common multiple} of the two divisors $D,D'$ are also defined componentwise by
	\begin{align*}
		\gcd(D, D') &= \sum_{P \in \mathcal{X}} \min(n_P, m_P)\, P, \\
		\operatorname{lcm}(D, D') &= \sum_{P \in \mathcal{X}} \max(n_P, m_P)\, P.
	\end{align*}
	
	These operations reflect natural notions of divisibility for divisors and are particularly useful in the study of divisors of functions and differentials.

	The \textbf{Riemann-Roch space} associated with a divisor $D$ is the $\F_q$-vector space
	\[
	\mathscr{L}(D) = \{ f \in \F_q(\mathcal{X}) \setminus \{0\} : (f) + D \ge 0 \} \cup \{0\}.
	\]
	This space consists of rational functions whose poles (if any) are bounded by $D$. 
	If $\deg(D)<0$ then $\ell(D)=0$. The dimension~$\ell(D) = \dim_{\F_q} \mathscr{L}(D)$ is finite and can be computed using the Riemann-Roch theorem as
	\[
	\ell(D) = \deg(D) - g + 1 + \ell(K - D),
	\]
	where $K$ is a canonical divisor on $\mathcal{X}$. 
	
	In the case where $\mathcal{X}$ is a smooth projective plane curve of order $d$, there is a classical and well-known description of a canonical divisor.
	Let $\ell_\infty$ be the line at infinity, and let
	\[
	D_\infty = \sum_{P \in C \cap \ell_\infty} I_P(C,\ell_\infty)\, P
	\]
	be the divisor on $C$ defined by the intersection multiplicities with $\ell_\infty$.
	Then the divisor $W=(d-3)D_\infty$ is a canonical divisor.
	
	An important consequence of the Riemann--Roch theorem is that
	\[
	\ell(D) = \deg(D) - g + 1
	\]
	whenever $\deg(D) > 2g - 2$. A divisor $D$ is said to be \textbf{special} if $\ell(K - D) > 0$, that is, if $K - D$ is linearly equivalent to an effective divisor. Equivalently, $D$ is special if
	\[
	\ell(D) > \deg(D) - g + 1.
	\]
	In contrast, $D$ is said to be \textbf{non-special} if $\ell(K - D) = 0$. An important result concerning special divisors is the \textbf{Clifford's Theorem}, which states that for any effective special divisor~$D$ on a on a projective, nonsingular curve of genus $g \geq 1$,
	\[
	\ell(D) \le \frac{\deg(D)}{2} + 1,
	\]
	with equality if and only if $D$ is canonical or linearly equivalent to a multiple of a degree-one divisor on a hyperelliptic curve.
	
	Let $D = \sum_{i=1}^R P_i$ be a divisor consisting of $R$ distinct $\F_q$-rational points $P_i \in \mathcal{X}(\F_q)$. Let $G$ be another $\F_q$-rational divisor on $\mathcal{X}$ such that $\supp(D) \cap \supp(G) = \varnothing$. Consider the \textbf{evaluation map}
	\[
	\begin{split}
		\mathrm{ev}_D:\quad \mathscr{L}(G) &\longrightarrow \F_q^R \\
		f &\longmapsto (f(P_1), \dots, f(P_R)).
	\end{split}
	\]
	This map is $\F_q$-linear, and it is injective whenever $R > \deg(G)$. The condition $\supp(D) \cap \supp(G) = \varnothing$ ensures that $f(P_i)$ is well-defined for all $f \in \mathscr{L}(G)$.
	
	The \textbf{functional AG code} $C_{\mathscr{L}}(D,G)$ is defined to be
	\[
	C_{\mathscr{L}}(D,G) = \mathrm{ev}_D(\mathscr{L}(G)) = \{ (f(P_1), \dots, f(P_R)) : f \in \mathscr{L}(G) \} \subseteq \F_q^R.
	\]
	This is an $\F_q$-linear code of length $R$, dimension $k_C = \ell(G) - \ell(G - D)$, and minimum distance $d_C \ge \bar{d} = R - \deg(G)$, where $\bar{d}$ is called the \textbf{designed minimum distance}. The \textbf{differential AG code} $C_{\Omega}(D,G)$ is defined using differential forms as
	\[
	C_{\Omega}(D,G) = \left\{ \left( \mathrm{res}_{P_1}(\omega), \dots, \mathrm{res}_{P_R}(\omega) \right) : \omega \in \Omega(G - D) \right\} \subseteq \F_q^R,
	\]
	where $\Omega(G - D)$ denotes the space of rational differentials $\omega$ on $\mathcal{X}$ such that $(\omega) \ge G - D$, and $\mathrm{res}_{P_j}(\omega)$ denotes the \textbf{residue} of $\omega$ at $P_j$. This is an $\F_q$-linear code with parameters $[R, R - \ell(G) + \ell(G - D), d']$, where
	\[
	d' \ge d^* = \deg(G) - 2g + 2,
	\]
	is called the \textbf{designed dual minimum distance}. It is known (see \cite{stichtenoth2009algebraic}) that
	\[
	C_{\Omega}(D,G) = C_{\mathscr{L}}(D,G)^{\perp} = C_{\mathscr{L}}(D, D - G + (\eta)),
	\]
	for a suitable differential $\eta$ with simple poles and residue $1$ at the points in $\supp(D)$.
	
	\begin{remark}
		Suppose $\deg(G) \le R$. Then the \textbf{Hamming defect} $\Hdef(C)$ of the AG code $C_{\mathscr{L}}(D,G)$ and its dual satisfies
		$
		\Hdef(C_{\mathscr{L}}(D,G)),\; \Hdef(C_{\mathscr{L}}(D,G)^\perp) \le g,
		$
		which shows that AG codes approach the Singleton bound with a defect bounded by the genus~$g$ of the underlying curve.
	\end{remark}

	Moreover, we have the following results concerning the interaction of algebraic geometry codes associated with different divisors. We omit the proofs, as they follow directly from the definitions of the corresponding divisors; the reader may refer, for instance, to \cite[Lemma 3.2]{bhowmick2024linear} for a detailed argument.
	
	\begin{proposition}
		\label{prop:basicsonGCDLCM}
		Let $C_L(D,G)$ and $C_L(D,H)$ be two algebraic geometry codes. Then:
		\begin{enumerate}
			\item $C_\L(D,\gcd(G,H)) \subseteq C_\L(D,G) \cap C_\L(D,H)$;
			\item $C_\L(D,G) + C_\L(D,H) \subseteq C_\L(D,\operatorname{lcm}(G,H))$.
		\end{enumerate}
	\end{proposition}

	\subsection{Elliptic Codes}
	
	Elliptic codes are algebraic geometry codes constructed from an elliptic curve $\mathcal{E}$, that is, a smooth projective curve of genus $g = 1$ defined over the finite field $\F_q$. An elliptic curve $\mathcal{E}$ over $\F_q$ is not only a projective algebraic curve of genus one with a distinguished $\F_q$-rational point serving as the group identity, but also an abelian variety. In particular, the set $\mathcal{E}(\F_q)$ of $\F_q$-rational points forms a finite abelian group, whose structure influences the construction and properties of the associated AG codes. For simplicity, we represent elliptic curves in Weierstrass form and take the group identity to be the point at infinity $P_\infty = (0:1:0)$.
	
	Elliptic $[R,k]$ codes exist for any parameters satisfying $R \leq N_q(1)$ and $1 \leq k \leq R - 1$, where~$N_q(1)$ denotes the maximum number of $\F_q$-rational points on a genus-one curve. By the Hasse--Weil bound, we have
	\[
	N_q(1) \le q + 1 + 2\sqrt{q}.
	\]
	
	The weight spectrum of an elliptic $[R,k]$ code $C$, expressed in terms of the coefficients $B_i$ of its weight enumerator, contains only one unknown value, namely $B_{k}$, corresponding to the number of codewords of minimum possible weight $k$. This allows for precise control of the structure of the code. The code $C$ is MDS if and only if $B_{k} = 0$. However, such cases are rare.
	
	\begin{theorem}
		Let $R > q + 1$ and let $C$ be an elliptic $\F_q$-$[R,k,d]$ code. If $q \geq 13$, then $C$ is not MDS, that is, $B_{k} > 0$, and hence $k + d < R + 1$.
	\end{theorem}
	
	Thus, MDS elliptic codes of length greater than $q + 1$ do not exist for any practical values of~$q$. The only exceptions occur for small $q$, where the group structure of the $\F_q$-rational points~$\mathcal{E}(\F_q)$ is highly restricted (that is, isomorphic to $\mathbb{Z}/2\mathbb{Z}$ or $(\mathbb{Z}/2\mathbb{Z})^2$). Indeed, almost all elliptic codes are not MDS but have Hamming defect equal to $1$. Elliptic codes that are not MDS are therefore \textbf{near MDS} (\textbf{NMDS}).
	
	Elliptic codes form a significant class of AG codes, lying between the trivial genus-zero case and more complex higher-genus constructions. Their formal self-duality, compact weight distribution, and frequent NMDS nature make them particularly attractive both from a theoretical and a practical perspective. However, their inability to attain the MDS bound for lengths greater than $q + 1$ imposes fundamental limitations, motivating the study of higher-genus codes for improved performance.
	
	\begin{definition}
		A (\textbf{rank-metric}) \textbf{code} is a subspace $\C \leq \Fq^{n \times m}$. The \textbf{minimum (rank) distance} of a non-zero code $\C$ is defined to be
		\[
		\drk(\C) = \min\{ \rk(M) : M \in \C,\ M \neq 0 \}.
		\]
		For $\C = \{0\}$, we define $\drk(\C) = n + 1$. The \textbf{dual code} of $\C$ is defined to be
		\[
		\C^\perp = \{ M \in \Fq^{n \times m} : \Tr(MN^{\textup{t}}) = 0 \text{ for all } N \in \C \},
		\]
		where $N^{\textup{t}}$ denotes the transpose of $N$ and $\Tr(A)$ denotes the trace of a square matrix $A$.
	\end{definition}
	
	\begin{notation}
		We refer to $\C$ as an $\Fq$-$[n \times m, k]$ code. When the minimum distance $d$ is specified, we refer to $\C$ as an $\Fq$-$[n \times m, k, d]$ code.
	\end{notation}
	
	It is well known that the dual code $\C^\perp$ is an $\Fq$-$[n \times m, nm - k]$ code. We recall the Singleton-type bound for rank-metric codes.
	
	\begin{theorem}[Rank-metric Singleton bound]
		\label{thm:ranksingl}
		Let $\C \leq \Fq^{n \times m}$ be an $\Fq$-$[n\times m,k,d]$ code. Then
		\[
		k \leq \max(m,n)(\min(m,n) - d + 1).
		\]
	\end{theorem}
	
	We say that $\C \leq \Fq^{n \times m}$ is an \textbf{MRD} (\textbf{maximum rank distance}) code if it meets the bound in Theorem~\ref{thm:ranksingl} with equality.

	\section{Tensor Rank Defect}\label{sec:genten}
	
	In this section, we introduce the notion of the \textit{tensor rank defect} of a rank-metric code and present some results on families of codes with tensor rank defect zero. We begin by recalling some preliminary definitions and results. A subspace of $\Fq^{n \times m}$ is said to be \textbf{perfect} if it is generated by rank-$1$ matrices. The \textbf{tensor rank} $\trk(\C)$ of $\C$ is defined to be the minimum dimension of a perfect subspace containing $\C$, that is,
	\begin{equation*}
		\trk(\C) = \min\{ \dim_{\F_q}(\A) : \A \leq \Fq^{n \times m} \text{ is perfect and } \C \leq \A \}.
	\end{equation*}
	This characterization follows from \cite[Proposition~14.45]{burgisser2013algebraic}; see also \cite[Proposition~3.4]{byrne2019tensor} and~\cite{byrne2023tensor}. In~\cite[Corollary~1]{kruskal1977three}, a bound on the tensor rank is given. We recall a coding-theoretic reformulation of this bound.
	
	\begin{theorem}[Kruskal's bound]
		Let $\C$ be an $\Fq$-$[n\times m,k,d]$ code.
		We have $\trk(\C) \ge k + d - 1.$
	\end{theorem}
	
	Following~\cite[Definition~4.12]{byrne2019tensor}, we say that $\C$ is a \textbf{Maximum Tensor Rank} (MTR) code if it meets the above bound with equality. We now introduce a new invariant for rank-metric codes.
	
	\begin{definition}[Tensor Rank Defect]
		We define the \textbf{Tensor Rank Defect} of $\C$ as
		\[
		\trkdef(\C) = \trk(\C) - (k + d - 1).
		\]
	\end{definition}
	
	In particular, MTR codes have tensor rank defect equal to $0$. In \cite{brockett1978optimal} (see also \cite[Chapter~18]{burgisser2013algebraic}), the authors establish a connection between linear block codes and tensor rank, which we briefly recall here. Let $V \in \FqnR$ and $W \in \FqmR$ be full-rank matrices, and define the linear map
	\begin{equation*}
		\phi_{V,W} : \Fq^R \longrightarrow \Fq^{n \times m}: x \longmapsto V \diag(x) W^{\textup{t}},
	\end{equation*}
	where $\diag(x)$ denotes the diagonal matrix whose $(i,i)$-entry is the $i$-th coordinate of $x$. Any matrix code can be realized as the image of a linear block code $C \leq \Fq^R$ under the map $\phi_{V,W}$ for a suitable choice of $V$ and $W$. Let $A_1, \ldots, A_R \in \Fq^{n \times m}$ be rank-$1$ matrices, and let $\A$ be the perfect subspace generated by $A_1, \ldots, A_R$. Define the map
	
	\begin{equation*}
		\psi_{\A}:\A\longrightarrow\Fq^R:\sum_{i=1}^R\mu_iA_i\longmapsto\sum_{i=1}^R\mu_ie_i,
	\end{equation*}
	where $\{e_1, \ldots, e_R\}$ is the standard basis of $\Fq^R$. By \cite[Lemma~4.13]{byrne2019tensor}, the map $\phi_{V,W}$ is invertible, and its inverse is $\psi_{\B}$, where $\B = \{ v_r \otimes w_r : r \in \{1, \ldots, R\} \}$, and $v_r$ and $w_r$ denote the $r$-th columns of $V$ and $W$, respectively. We recall the following result.

	\begin{proposition}[{\cite[Corollary~4.14]{byrne2019tensor}}]
		\label{prop:siaga}
		Suppose that $\trk(\C)=R$. Let $D_1, \ldots, D_k \in \Fq^{R \times R}$ be diagonal matrices such that
		\[
		\C = \langle V D_1 W^{\textup{t}}, \ldots, V D_k W^{\textup{t}} \rangle_{\F_q}
		\]
		for some full-rank matrices $V \in \Fq^{n \times R}$ and $W \in \Fq^{m \times R}$. The following hold:
		\begin{enumerate}
			\item \label{item1:siaga} $\rk(M) \leq \wt\big(\phi_{V,W}^{-1}(M)\big)$ for every $M \in \C$.
			\item $\phi_{V,W}^{-1}(\C)$ is an $\Fq$-$[R, k, \geq d]$ block code.
			\item If $\C$ is MTR, then the linear block code $\phi_{V,W}^{-1}(\C)$ is MDS.
		\end{enumerate}
	\end{proposition}
	
	In the remainder, we denote by $C_V$ and $C_W$ the linear block codes generated by the rows of the matrices $V$ and $W$, respectively.
	
	\begin{notation}
		For any $v, w \in \Fq^R$, we denote by $*$ the \textbf{Schur product} (or componentwise product) of $v$ and $w$, that is, $v * w = (v_1 w_1, \ldots, v_R w_R)$. Moreover, for a code $C \leq \Fq^R$, we define $C * v = \{ c * v : c \in C \}$.
	\end{notation}
	
	\begin{proposition}\label{prop:dimmatC}
		We have that $\phi_{V,W}(C)$ is an $\Fq$-$[n \times m, k', d']$ code satisfying
		\[
		\trk(\phi_{V,W}(C)) \leq R, \quad k' \leq k_C, \quad\text{ and }\quad d' \leq d_C.
		\]
		Furthermore, we have $k' = k_C$ if and only if $C_W * c \not\subseteq C_V^\perp$ for all $c \in C \setminus \{0\}$.
	\end{proposition}
	
	\begin{proof}
		The inequality $d' \leq d_C$ follows from Proposition~\ref{prop:siaga}\eqref{item1:siaga}. We have that $k' \leq k_C$ since~$\phi_{V,W}$ is an $\Fq$-linear homomorphism from $C$ into $\phi_{V,W}(C)$. Moreover, $k' = k_C$ if and only if $\ker \phi_{V,W} \cap C =\{c \in C: V_c W^\textup{t}=0\}$ is trivial,  
		which is equivalent to requiring that $C_W *c$ is not contained the null-space of $V$. 
	\end{proof}
	
	\begin{notation}
		Given a matrix $A\in \Fq^{r \times s}$ and $D=\diag(x) \in \Fq^{s\times s}$ for some $x \in \Fq^s$ we write $A_x$ to denote the $r \times s$ matrix $AD$.
	\end{notation}
	
	The following result establishes a relationship between the tensor rank defect and the Hamming defect.
	
	\begin{proposition}
		\label{prop:defect}
		If $\C=\phi_{V,W}(C)$ then the following hold.
		\begin{enumerate}
			\item\label{item1:defects} $\trkdef(\C) \leq \Hdef(C) + k_C + d_C - k - d$.
			
			\item\label{item2:defects} If $C_W * c \nsubseteq C_V^\perp$ for every  $c \in C\setminus\{0\}$, then
			\[
			\trkdef(\C) 
			\leq \Hdef(C) + d_C - d
			= \Hdef(C) + d_C - \min\{ \dim(C_W * c) : c \in C,\ c \neq 0 \}.
			\]
		\end{enumerate} 
	\end{proposition}
	\begin{proof}
		Part~\eqref{item1:defects} follows from the definitions of $\trkdef(\C)$ and $\Hdef(C)$. If $C_W * c \nsubseteq C_V^\perp$ for every $c \in C \setminus \{0\}$, then $V (W_c)^{\textup{t}} \neq 0$ for every $c \in C \setminus \{0\}$. Therefore, $\phi_{V,W}$ is injective on $C$, and hence an isomorphism from $C$ onto $\phi_{V,W}(C) = \C$. In particular, $k_C = \dim(\C)$. Part~\eqref{item2:defects} follows from Part~\eqref{item1:defects} and the fact that
		\[
		d \leq \rk\big(V (W_c)^{\textup{t}}\big) = \dim(C_W * c) - \dim\big(C_V^\perp \cap (C_W * c)\big) = \dim(C_V * c) - \dim\big(C_W^\perp \cap (C_V * c)\big),
		\]
		for all $c \in C \setminus \{0\}$, by \cite[Lemma~5.2]{byrne2019tensor}.
	\end{proof}

	\begin{remark}
		We note the following equivalences.
		\[
		C_W * c \subseteq C_V^\perp 
		\;\Longleftrightarrow\;
		C_W \subseteq (C_V * c)^\perp
		\;\Longleftrightarrow\;
		C_V \subseteq (C_W * c)^\perp
		\;\Longleftrightarrow\;
		C_V * c \subseteq C_W^\perp.
		\]
	\end{remark}
	
	The following examples illustrate Proposition~\ref{prop:defect}.

	\begin{example}
		Consider the following matrices with entries in $\F_3$:
		\begin{equation*}
			V = \begin{pmatrix}
				1 & 0 & 0 & 2 & 0 & 2 \\
				0 & 1 & 0 & 1 & 0 & 1 \\
				0 & 0 & 0 & 0 & 1 & 2
			\end{pmatrix}, \qquad
			G = \begin{pmatrix}
				1 & 0 & 2 & 2 & 0 & 2 \\
				0 & 1 & 0 & 0 & 0 & 2
			\end{pmatrix}, \qquad
			W = \begin{pmatrix}
				1 & 0 & 2 & 0 & 2 & 2 \\
				0 & 1 & 1 & 0 & 2 & 2 \\
				0 & 0 & 0 & 1 & 1 & 0
			\end{pmatrix}.
		\end{equation*}
		Let $C_V$, $C$, and $C_W$ be the linear block codes generated by the matrices $V$, $G$, and $W$, respectively. In particular, $C_V$ is an $\F_3$-$[6,3,2]$ linear block code, $C$ is an $\F_3$-$[6,2,2]$ linear block code, and $C_W$ is an $\F_3$-$[6,3,2]$ linear block code. Let $\C = \phi_{V,W}(C)$. Then $\C$ has dimension $2$ and is generated by 
		\begin{equation*}
			\begin{pmatrix}
				0&1&0\\0&0&0\\2&1&0
			\end{pmatrix} \qquad\textup{ and }\qquad
			\begin{pmatrix}
				0&0&1\\1&2&2\\1&2&1
			\end{pmatrix}.
		\end{equation*}
		One can check that, for every nonzero $c \in C$, we have $C_W * c \nsubseteq C_V^\perp$. Therefore, by Proposition~\ref{prop:defect}, we get $\Delta^\trk\leq \Delta^\H+d_C-d=3$.
	\end{example}
	
	\begin{example}
		Consider the following matrices with entries in $\F_5$.
		\begin{equation*}
			V=
			\begin{pmatrix}
				1&0&0&4&3\\0&1&0&2&3\\0&0&1&1&1
			\end{pmatrix},\qquad
			G=
			\begin{pmatrix}
				1&0&0&1&4\\0&1&0&4&4\\0&0&1&1&2
			\end{pmatrix}\qquad
			W=
			\begin{pmatrix}
				1&0&0&3&2\\0&1&0&1&2\\0&0&1&3&4
			\end{pmatrix}
		\end{equation*}
		
		Let $C_V$, $C$, and $C_W$ be the linear block codes generated by the matrices $V$, $G$, and $W$, respectively. One can check that they are all $\F_5$-linear $[5,3,3]$ codes and hence are MDS. Let~$\C = \phi_{V,W}(C)$. Then $\C$ is an $\F_5$-linear $[3 \times 3, 3, 3]$ code generated by
		\begin{equation*}
			\begin{pmatrix}
				1&0&0\\0&3&0\\0&0&2
			\end{pmatrix},\qquad
			\begin{pmatrix}
				0&1&0\\0&0&3\\2&3&0
			\end{pmatrix},\qquad
			\begin{pmatrix}
				0&0&1\\3&2&2\\0&2&4
			\end{pmatrix}.
		\end{equation*}
		It can be checked that for every nonzero $c \in C$, we have $C_W * c \nsubseteq C_V^\perp$. Therefore, by Proposition~\ref{prop:defect}, we get $\trkdef(\C) \leq \Hdef(C) + d_C - d = 0 + 3 - 3 = 0$. Hence, $\C$ is MTR. A basis of rank-$1$ matrices for the smallest perfect subspace containing $\C$ is given by $\B = \{ v_r \otimes w_r : r \in \{1, \ldots, 5\} \}$, namely, 
		\begin{equation*}
			\left\{\begin{pmatrix}
				1&0&0\\0&0&0\\0&0&0
			\end{pmatrix},
			\begin{pmatrix}
				0&0&0\\0&1&0\\0&0&0
			\end{pmatrix},
			\begin{pmatrix}
				0&0&0\\0&0&0\\0&0&1
			\end{pmatrix},
			\begin{pmatrix}
				2&4&2\\1&2&1\\3&1&3
			\end{pmatrix},
			\begin{pmatrix}
				1&1&2\\1&1&2\\2&2&4
			\end{pmatrix}\right\}.
		\end{equation*}
	\end{example}

	\begin{proposition}\label{prop:MTRchar}
		Let $\C = \phi_{V,W}(C)$ be an $\Fq$-$[n \times m, k, d]$ code of tensor rank $R$, and suppose that $\trkdef(\C) = 0$.  
		For every $c \in C$ of Hamming weight $d$, the following hold:
		\begin{enumerate}
			\item $\dim(C_W * c) = \dim(C_V * c) = d$,
			\item $(C_W * c) \cap C_V^\perp = (C_V * c) \cap C_W^\perp = \{0\}$,
			\item $\dim\big(C_V \cap (C_W * c)^\perp\big) = n - d$,
			\item $\dim\big(C_W \cap (C_V * c)^\perp\big) = m - d$.
		\end{enumerate}
	\end{proposition}
	
	\begin{proof}
		Since $\C$ is MTR, by Proposition~\ref{prop:siaga} we have that $\dim(C) = k$, $d^H(C) = d$, and that $C$ is MDS. Let $c \in C$ have Hamming weight $d$. Then $d=\rk(V(W_c)^{\textup{t}})=\dim(C_W * c)-\dim((C_W*c) \cap C_V^\perp)$. Since $\dim(C_W * c) \leq d$, it follows that $\dim(C_W * c) = d$ and $\dim\big((C_W * c) \cap C_V^\perp\big) = 0$. Similarly, using $d = \rk(W (V_c)^{\textup{t}})$, we obtain $\dim(C_V * c) = d$ and $(C_V * c) \cap C_W^\perp = \{0\}$. To prove (3) and (4), observe that $d=\rk(V(W_c)^{\textup{t}})=\rk(V)-\dim((C_W*c)^\perp \cap C_V)=n-\dim((C_W*c)^\perp \cap C_V)=\rk(W)-\dim((C_V*c)^\perp \cap C_W)=m-\dim((C_V*c)^\perp \cap C_W)$.
	\end{proof}
	
	As a conseqeunce, we have the following result.
	
	\begin{corollary}\label{cor:VWMDS}
		Let $\C = \phi_{V,W}(C)$ be an $\Fq$-$[n \times m, k, d]$ code of tensor rank $R$, and suppose that $\trkdef(\C) = 0$. If $d = n$, then $C_V$ is MDS. If $d = m$, then $C_W$ is MDS.
	\end{corollary}
	
	\begin{proof}
		For each $c \in C$ of Hamming weight $d$, define
		$\Pi_c : \Fq^R \longrightarrow \Fq^{R-d}: x \longmapsto (x_i)_{i \notin \supp(c)}$. Since $\C$ is MTR, by Proposition~\ref{prop:MTRchar}(1), for every such $c$ we have
		\[
		\dim\big(\Pi_c(C_W)\big) = d = \dim\big(\Pi_c(C_V)\big).
		\]
		Since $C$ is MDS, every subset of size $d$ is the support of a codeword of $C$, and hence every such subset is an information set for both $C_V$ and $C_W$. Therefore, if $d = n$, then $C_V$ is MDS, and if $d = m$, then $C_W$ is MDS.
	\end{proof}
	
	\begin{remark}
		Clearly, from the proof of Corollary~\ref{cor:VWMDS}, it follows that if  $\C = \phi_{V,W}(C)$ has tensor rank $R$ and is MTR, then both $C_V^\perp$ and $C_W^\perp$ have minimum distance at least $d + 1$.
	\end{remark}
	
	\begin{corollary}
		\label{cor:defect}
		We have $\trkdef(\C) \leq \Hdef(C) + d_C - d$ if one of the following conditions holds:
		\begin{enumerate}
			\item $|\supp(C) \cap \supp(C_W)| \leq d^H(C_V^\perp) - 1$,
			\item $|\supp(C) \cap \supp(C_V)| \leq d^H(C_W^\perp) - 1$.
		\end{enumerate}
	\end{corollary}
	\begin{proof}
		For any $c \in C$, we have $\supp(C_W * c) \subseteq \supp(C) \cap \supp(C_W)$. Therefore, if $|\supp(C) \cap \supp(C_W)| \leq d^H(C_V^\perp) - 1$, then every nonzero vector in $C_W * c$ has support of size strictly less than $d^H(C_V^\perp)$, and hence $C_W * c \nsubseteq C_V^\perp$ for all nonzero $c \in C$. Analogously, the second condition implies that $C_V * c \nsubseteq C_W^\perp$ for all nonzero $c \in C$. The statement now follows from Proposition~\ref{prop:defect}.
	\end{proof}
	
	\begin{example}
		Consider the following matrices with entries in $\F_5$.
		\begin{equation*}
			V=
			\begin{pmatrix}
				1&0&0&1&2&0\\0&1&0&2&1&0\\0&0&1&4&1&0
			\end{pmatrix},\qquad
			G=
			\begin{pmatrix}
				1&0&4&4&3&4\\0&1&4&1&1&2
			\end{pmatrix}\qquad
			W=
			\begin{pmatrix}
				1&0&0&0&0&4\\0&1&0&0&0&2\\0&0&1&0&0&2\\0&0&0&1&0&4\\0&0&0&0&1&4
			\end{pmatrix}
		\end{equation*}
		Let $C_V$, $C$, and $C_W$ be the linear block codes generated by the matrices $V$, $G$, and $W$, respectively. One can check that $C_V$ is an $\F_5$-$[6,3,3]$ linear block code, $C$ is an $\F_5$-$[6,2,5]$ MDS linear block code, and $C_W$ is an $\F_5$-$[6,5,2]$ linear block code. Moreover, one can check that $C_W^\perp$ is the $\F_5$-$[6,1,6]$ code generated by $(1,3,3,1,1,1)$, and that $C_V * c \nsubseteq C_W^\perp$ for all nonzero $c \in C$. Therefore, by Corollary~\ref{cor:defect}, we obtain $\trkdef(\C) \leq \Hdef(C) + d_C - d = 0 + 5 - 3 = 2$.
	\end{example}
	
	\subsection{A construction based on Cauchy codes}
	
	We now present the main result of this section, in which we extend the construction of MTR codes given in \cite[Theorem~5.7]{byrne2019tensor} to a strictly larger class of codes. 
	We first introduce some notation.
	
	For any positive integer $s$, we let $\Fq[x]_{\leq s}$ denote the set of
	polynomials in $\Fq[x]$ of degree at most $s$. We let $\Fq[x,y]_s$ denote the space of homogeneous polynomials of degree $s$ in two variables $x$ and $y$, together with the zero polynomial.
	
	Let $\mathbb{P}^1(\Fq)$ denote the projective line over $\Fq$. For any $f \in \Fq[x,y]_s$ we define the map $f:\mathbb{P}^1(\Fq) \longrightarrow \Fq$, by 
	\[
	f(P) = \left\{ 
	\begin{array}{ll}
		f(P,1) & \text{ if }P \in \Fq\\
		f(1,0) & \text{ if }P \text{ is the point at infinity of }\mathbb{P}^1(\Fq). 
	\end{array}
	\right.        
	\]
	Furthermore, for any $f\in \Fq[x,y]_s$, we write $f_x \in \Fq[x]_{\leq s}$ to denote its dehomogenization with respect to $y$ and we write $f_y \in \Fq[y]_{\leq s}$ to denote its dehomogenization with respect to $x$, that is, $f_x(x)=f(x,1)$ and $f_y(y)=f(1,y)$. 
	
	For any $\bar{P}=(P_1, \ldots, P_R) \in (\mathbb{P}^1(\Fq))^R$, we define the evaluation map
	\[
	\mathrm{ev}_{\bar{P}}: \Fq[x,y]_s \longrightarrow \Fq^R, \mathrm{ev}_{\bar{P}}(f)=(f(P_1),\dots f(P_R)) \text{ for all } f \in  \Fq[x,y]_s
	\]
	Given $\beta=(\beta_1,\dots,\beta_R) \in \Fq^R$ and $\bar{P}=(P_1, \ldots, P_R) \in (\mathbb{P}^1(\Fq))^R$, with the $P_i$ pairwise distinct, the Cauchy code $C_k(\bar{P},\beta)$ is the $\Fq$-$[R,k]$ code defined by
	$$  C_k(\bar{P},\beta)=\{\beta* \mathrm{ev}_{\bar{P}}(f): f \in \Fq[x,y]_{k-1} \}.$$

	Let $0<s\le t\in\mathbb{N}$. Define the map
	\[
	\Phi_{s,t}: \Fq[x,y]_{s}\times\Fq[x,y]_{t} \longrightarrow \Fq[x,y]_{s+t}:
	(f,h)\longmapsto f\cdot h
	\]
	define the set 
	$$S_{s,t}=\left\{p\in\Fq[x,y]_{s+t}: \textup{ if } f\in\Fq[x]_{\leq s}  \text{ divides } p_x,\text{ then } \deg\left(\frac{p_x}{f}\right) \geq t+1\right\} \subseteq \Fq[x,y]_{s+t}.$$
	\begin{lemma}\label{lem:Sst}
		Let $0<s\le t\in\mathbb{N}$. Then
		$S_{s,t}=\Fq[x,y]_{s+t}\setminus\mathrm{Im}~\Phi_{s,t}.$
	\end{lemma}
	\begin{proof}
		Let $p \in \textup{Im} \Phi_{s,t}$. Then there exist
		$f \in \Fq[x,y]_s$ and $h \in \Fq[x,y]_t$ such that
		$p=fh$ so and $p_x=f_xh_x$ with $\deg(f_x)\leq s$ and 
		$\deg(h_x)\leq t$. Then $p \notin S_{s,t}$ and hence we deduce that 
		$S_{s,t}\subseteq \Fq[x,y]_{s+t}\setminus\mathrm{Im}~\Phi_{s,t}.$
		Conversely, suppose that $p \in \Fq[x,y]_{s+t}$ such that
		$p_x = fh$ for some $f \in \Fq[x]_{\leq s}$ and $h\in \Fq[x]_{\leq t}$. Let $p_x$ have degree $\ell \leq s+t$.
		
		Let $\ell = s+t-r$, for some $r\geq 0$. 
		Then $p=y^{s+t-r}p'$, where $p' \in \Fq[x,y]_{\ell}$ such that $(p')_x = p_x$. Let $\deg(f)=s'$, let $\deg(h)=t'$
		and let $f' \in \Fq[x,y]_{s'}$, $h' \in \Fq[x,y]_{t'}$ such 
		that $(f')_x=f$ and $(h')_x=h$. That is, let $f'$ and $h'$
		be homogenizations of $f$ and $h$, respectively.
		We have $s'+a= s$ and $t' +b = t$ for some $a,b \leq 0$. 
		Since $p' = f' h'$, we have 
		$p = y^{s+t-r}p' = y^{s+t-r}f'h'$, and hence
		$s+t = s+t-r+s'+t'$, which gives $r=s'+t'$.
		Then $p =y^{a+b}f'h' = (y^a f')(y^b h')$ and 
		$y^a f' \in \Fq[x,y]_s$ and $y^b h' \in \Fq[x,y]_t$,
		from which we deduce that $p \in \textup{Im}~\Phi_{s,t}$.\end{proof}
	
	The identification of this set leads to the following result, which generalizes the construction in \cite[Theorem~5.7]{byrne2019tensor}. As in that case, the rank-metric code produced here is a constant rank MRD code with tensor rank defect zero.
	
	\begin{theorem}
		Let $0 < d < k < R$ be positive integers satisfying $R = k + d - 1$, and let $\bar{P}=(P_1, \ldots, P_R) \in (\mathbb{P}^1(\Fq))^R$, with the $P_i$ pairwise distinct.
		Let $f(x,y) \in \Fq[x,y]_{k}$. Let $C = C_k(\bar{P}, \mathbf{1})$, let 
		$V \in \Fq^{k \times R}$ be a parity-check matrix for 
		$C_{R-k}(\bar{P},\mathrm{ev}_{\bar{P}}(f))$, and let 
		$W \in \Fq^{d \times R}$ be a generator matrix for $C_d(\bar{P}, \mathbf{1})$. 
		Then $\C = \phi_{V,W}(C)$ is MTR if and only if $f \in S_{d-1,k-1}$.
	\end{theorem}
	\begin{proof}
		We show that for every $c \in C$, we have
		$C_V^\perp \cap (C_{W}*c) = \{0\}.$
		For any $c \in C$, we have $c = \mathrm{ev}_{\bar{P}}(h)$ for some nonzero $h \in \F_q[x,y]_{k-1}$. If $b \in C_V^\perp \cap (C_{W}*c)$ for some $c$, then there exist $\mu \in \F_q[x,y]_{R-k-1} = \F_q[x,y]_{d-2}$ and $\lambda \in \F_q[x,y]_{d-1}$ such that
		\[
		b = \mathrm{ev}_{\bar{P}}(f)\cdot \mathrm{ev}_{\bar{P}}(\mu) = \mathrm{ev}_{\bar{P}}(h)\cdot \mathrm{ev}_{\bar{P}}(\lambda).
		\]
		Since $b$ is the evaluation of a polynomial in $\Fq[x,y]_{R-1}$, this holds if and only if $f\mu - h\lambda = 0$ in $\F_q[x,y]$. It follows that $C_V^\perp \cap (C_{W}*c) = \{0\}$ for all $c \in C$ if and only if, for all $\mu \in \F_q[x,y]_{R-k-1}$, $\mu f \notin \mathrm{Im}(\Phi_{d-1,k-1})$.  
		From Lemma \ref{lem:Sst}, this holds if and only if $\mu f \in S_{d-1,k-1}$ for all $\mu\in \F_q[x,y]_{d-2}$. In particular,
		this must hold for $\mu=y^{d-2}$. If $y^{d-2} f \in S_{d-1,k-1}$ then
		for any divisor $a \in \Fq[x]_{\leq d-1}$ of $(y^{d-2}f)_x = f_x$ we have that $\deg\left(\frac{f_x}{a}\right) \geq k$, which holds if and only if $f \in S_{d-1,k-1}$. 
		If $C_V^\perp \cap (C_{W}*c)$ is trivial for all $c \in C$, we also have that $\phi_{V,W}$ is an isomorphism. Furthermore, since the code $C_W$ is MDS, we have that for any $c\in C$, that 
		$\rk(W_c) = \min(\rk(W),d^H(c)) = d$ and thus $\C$ is an $\Fq$-$[k \times d, k,d]$ rank-metric code.
		Hence, $\C$ is MTR if and only if $f \notin \mathrm{Im}(\Phi_{d-1,k-1})$. 
	\end{proof}
	
	We conclude this section with the following example.
	\begin{example}
		Let $d = 3$, $k = 6$, and $R = 8$. Let $\omega$ be a root of $x^3 + x + 2$ over $\F_3$, and let~$\alpha = (\omega^i)_{0 \leq i \leq 7}$. Then $\omega$ is a primitive element of $\F_9$, and hence the entries of $\alpha$ are pairwise distinct. Let
		\[
		f(x,y) = (x^3 + x y^2 + \omega y^3)^2 = x^6 + 2x^4 y^2 + \omega^5 x^3 y^3 + x^2 y^4 + \omega^5 x y^5 + \omega y^6 \in \F_9[x,y].
		\]
		Let $C = C_6(\alpha, \mathbf{1})$. A parity-check matrix $V$ of~$C_3(\alpha, \mathrm{ev}_\alpha(f))$ and generator matrix $W$ of~$C_3(\alpha, \mathbf{1})$, are given by
		\[
		V=\left[ 
		\begin{array}{cccccccc}
			1 & 0 & 0 & 0 & 0 & 0 & \omega^5 & \omega^6  \\
			0 & 1 & 0 & 0 & 0 & 0 & \omega^5 & 1  \\
			0 & 0 & 1 & 0 & 0 & 0 & \omega^7 & \omega  \\
			0 & 0 & 0 & 1 & 0 & 0 & 2 & \omega  \\
			0 & 0 & 0 & 0 & 1 & 0 & \omega^2 & \omega^6  \\
			0 & 0 & 0 & 0 & 0 & 1 & \omega & \omega^7  \\
		\end{array}
		\right],\qquad
		W=\left[ 
		\begin{array}{cccccccc}
			1 & 0 & 0 &  \omega^3 & \omega   & 1        & \omega   & 2\\
			0 & 1 & 0 &  \omega^3 & 1        & \omega^2 & \omega^2 & \omega^7\\
			0 & 0 & 1 &  \omega^6 & \omega^5 & \omega^6 & \omega   & \omega^5
		\end{array}
		\right]. 
		\]
		The rank-metric code $\C = \phi_{V,W}(C)$ is an $\F_9$-linear $[6 \times 3, 6, 3]$ MRD matrix code of tensor rank $8$, and hence is MTR.
	\end{example}

	\section{AG code constructions}
	\label{sec:intersection}
	
	The tensor rank defect is closely tied to the interaction between the auxiliary codes $C_V$, $C_W$, and the inner code $C$. The previous results highlight that controlling intersections of the form
	\[
	C_W * c \cap C_V^\perp \quad \text{and} \quad C_V * c \cap C_W^\perp
	\]
	is crucial for obtaining rank-metric codes with small tensor rank defect, and in particular for ensuring the MTR property. Motivated by this, we now turn our attention to constructing families of codes where these intersections can be explicitly controlled or forced to be trivial.
	
	A natural and powerful framework for this purpose is provided by algebraic geometry (AG) codes. Indeed, AG codes form a well-known and far-reaching generalization of Reed--Solomon codes, which already played a central role in the constructions discussed earlier. Their rich algebraic structure offers additional flexibility, allowing us to tailor both the parameters and the multiplicative behavior of codewords. In the following, we investigate intersections of AG codes, with the goal of deriving explicit conditions under which the dimension of such intersections is zero or small. This approach enables us to extend the previous constructions and to obtain new families of rank-metric codes with good estimates on the tensor rank. For simplicity, we will denote $C_{\mathscr{L}}(D,G)$ by $C(D,G)$.
	
	\begin{theorem}\label{thm:generalcurves}
		Let $\mathcal{P}$ be a proper subset of $\mathcal{X}(\Fq)$ and set $D = \sum_{P \in \mathcal{P}} P$.  
		Let $G,H,J$ be divisors on $\mathcal{X}$ of degree at most $\deg(D)-1$ such that
		\[
		(\supp(G) \cup \supp(H) \cup \supp(J)) \cap \supp(D) = \varnothing.
		\]
		For each $f \in \L(J)$, let $H_f=H - (f)$.
		Then, for any $f \in \mathscr{L}(J)$, we have
		\[
		\ell\big(\gcd(G,\, H_f\big) \leq 
		\dim\left( C(D,G) \cap \big(\mathrm{ev}_D(f) * C(D,H)\big) \right)
		\le 
		\ell\big(\gcd(G,\, H_f\big) + \ell\big(\mathrm{lcm}(G,\, H_f) - D\big).
		\]
	\end{theorem}

	\begin{proof}
		For every $f \in \mathscr{L}(J)$. 
		We have that $\mathrm{ev}_D(f) * C(D,H) = C(D,H_f)$ and the map $\lambda \mapsto f \lambda$ defines an isomorphism of vector spaces between $\mathscr{L}(H)$ and $\mathscr{L}(H_f)$. 
		From Proposition \ref{prop:basicsonGCDLCM}, for any $f \in \mathscr{L}(J)$ we have that 
		\[
		C(D,\gcd(G,H_f))\subseteq C(D,G)\cap C(D,H_f),
		\]
		and since $\deg(G),\deg(H_f) <\deg(D)$, we get that 
		\[
		\ell(\gcd(G,H_f))=\dim(C(D,\gcd(G,H_f)))\le \dim(C(D,G)\cap C(D,H_f)).
		\]

		To establish the other inequality,
		we construct a vector space $U$ consisting of all valid pairs of functions that evaluate to the same codeword on $D$:
		\[
		U = \big\{ (\mu, \nu) \in \mathscr{L}(G) \times \mathscr{L}(H_f) \;\big|\; \mu - \nu \in \mathscr{L}\big(\mathrm{lcm}(G, H_f) - D\big) \big\}.
		\]
		
		Consider the linear map $\varphi : U \to \mathscr{L}\big(\mathrm{lcm}(G, H_f) - D\big)$ defined by $\varphi(\mu, \nu) = \mu - \nu$.
		A pair $(\mu, \nu)$ belongs to $\ker(\varphi)$ if and only if $\mu = \nu$. Since $\mu \in \mathscr{L}(G)$ and $\nu \in \mathscr{L}(H_f)$, this requires 
		\[
		\mu \in \mathscr{L}(G) \cap \mathscr{L}(H_f) = \mathscr{L}\big(\gcd(G, H_f)\big).
		\]
		Thus, $\ker(\varphi) \cong \mathscr{L}\big(\gcd(G, H_f)\big)$, which gives
		$
		\dim(\ker(\varphi)) = \ell\big(\gcd(G, H_f)\big).
		$
		
		The image of $\varphi$ is a subspace of $\mathscr{L}\big(\mathrm{lcm}(G, H_f) - D\big)$, so its dimension is strictly bounded from above by $\ell\big(\mathrm{lcm}(G, H_f) - D\big)$. 
		The Rank-Nullity Theorem yields that
		\[
		\dim(U) \le \ell\big(\gcd(G, H_f)\big) + \ell\big(\mathrm{lcm}(G, H_f) - D\big).
		\]
		
		Now, consider the evaluation map $E : U \to \mathbb{F}_q^n$ defined by $E(\mu, \nu) = \mathrm{ev}_D(\mu)$. 
		We claim that the image of $E$ is exactly $C(D,G) \cap C(D,H_f)$.
		
		First, let $(\mu, \nu) \in U$. By definition, $\mu \in \mathscr{L}(G)$, so $\mathrm{ev}_D(\mu) \in C(D,G)$. Furthermore, since $\mu - \nu \in \mathscr{L}\big(\mathrm{lcm}(G, H_f) - D\big)$, the difference $\mu - \nu$ vanishes at all points in $\supp(D)$. Thus, $\mathrm{ev}_D(\mu) = \mathrm{ev}_D(\nu)$. Since $\nu \in \mathscr{L}(H_f)$, we have $\mathrm{ev}_D(\nu) \in C(D,H_f)$, meaning $E(\mu, \nu) \in C(D,G) \cap C(D,H_f)$.
		
		Conversely, let $c \in C(D,G) \cap C(D,H_f)$. There exist $\mu \in \mathscr{L}(G)$ and $\nu \in \mathscr{L}(H_f)$ such that $c = \mathrm{ev}_D(\mu) = \mathrm{ev}_D(\nu)$. Since $\mu$ and $\nu$ both evaluate to the same vector on $D$, the difference $\gamma = \mu - \nu$ must vanish on $D$. 
		Since we have $(\gamma) \ge -\mathrm{lcm}(G, H_f)$ 
		we obtain $\gamma \in \mathscr{L}\big(\mathrm{lcm}(G, H_f) - D\big)$. Therefore, $(\mu, \nu) \in U$ and $E(\mu, \nu) = c$.
		
		Since $C(D,G) \cap C(D,H_f)$ is the image of the linear map $E$ defined on $U$, we have 
		\[
		\dim\left( C(D,G) \cap \big(\mathrm{ev}_D(f) * C(D,H)\big) \right) \le \dim(U).
		\]
		
		Substituting into our upper bound for $\dim(U)$, we get that
		\[
		\dim\left( C(D,G) \cap \big(\mathrm{ev}_D(f) * C(D,H)\big) \right)
		\le 
		\ell\big(\gcd(G,\, H - (f))\big)
		+ \ell\big(\mathrm{lcm}(G,\, H - (f)) - D\big),
		\]
		which concludes the proof.
	\end{proof}

	\begin{remark}
		The upper bound in Theorem~\ref{thm:generalcurves} is sharp in general, even for curves of genus zero. In particular, there exist choices of divisors $G,H,J$ and evaluation sets $D$ for which the inequality in Theorem~\ref{thm:generalcurves} is attained with equality. 
		To emphasize this point, we provide below an explicit example in which the bound is met exactly, focusing for simplicity on a situation involving only the divisors $G$ and $H$.
	\end{remark}
	
	\begin{example}
		Let $\mathcal{X}=\mathbb{P}^1_{\mathbb{F}_7}$, realized as the projective line $\mathcal{X} : 3x+y-z=0 \subset \mathbb{P}^2_{\mathbb{F}_7}.$
		Let $P_1=(5:0:1), P_2=(0:1:1)$, and let
		\[
		G = 5P_1, 
		\qquad 
		H = 5P_2.
		\]
		Let $D$ be the sum of all $\mathbb{F}_7$-rational points of $\mathcal{X}$ not contained in the union of the supports of $G$, $H$ and $J$, that is, $
		D = \sum_{P \in \mathcal{X}(\mathbb{F}_7)\setminus\{P_1,P_2\}} P. $
		Then $\deg(D)=6$ and $(\supp(G)\cup\supp(H)\cup\supp(J)) \cap \supp(D) = \varnothing .$
		We have that
		\[
		\ell\big(\gcd(G,H)\big)=1,
		\qquad
		\ell\big(\mathrm{lcm}(G,H)-D\big)=\ell(5P_1+5P_2-D)=5.
		\]
		A direct computation using AG codes over $\mathbb{F}_7$ yields
		$
		\dim\!\left( C(D,G)\cap C(D,H) \right)=6.
		$
		Therefore,
		\[
		\dim\!\left( C(D,G)\cap C(D,H) \right)
		=
		\ell\big(\gcd(G,H)\big)
		+
		\ell\big(\mathrm{lcm}(G,H)-D\big),
		\]
		showing that equality holds in Theorem~\ref{thm:generalcurves}. This demonstrates that the upper bound of Theorem~\ref{thm:generalcurves} is sharp even for rational curves.
	\end{example}
	
	\begin{remark}
		The example above shows that the contribution of the term 
		\(\ell\big(\mathrm{lcm}(G,H-(f))-D\big)\) 
		in Theorem~\ref{thm:generalcurves} can be essential in order to achieve equality in the bound. 
		This term accounts for the possible presence of nonzero functions in 
		\(\mathscr{L}(\mathrm{lcm}(G,H-(f)))\) that vanish on the evaluation set~$D$, reflecting a potential failure of injectivity of the evaluation map on this Riemann-Roch space.

		At the same time, it is not necessary that this term vanishes for the intersection to be trivial. 
		Indeed, as we will show in further examples, it may happen that
		\[
		\ell\big(\mathrm{lcm}(G,H-(f))-D\big)>0
		\quad\text{while}\quad
		C(D,G)\cap \big(\mathrm{ev}_D(f)*C(D,H)\big)=\{0\}.
		\]
		In such cases, although nonzero functions vanishing on $D$ exist in 
		\(\mathscr{L}(\mathrm{lcm}(G,H-(f)))\), none of them arise from elements of
		\(\mathscr{L}(G)\cap f\mathscr{L}(H)\).
		Thus, the lcm term measures a possible contribution to the intersection rather than an unavoidable one.
		
		Consequently, while the lcm term is indispensable for obtaining a uniform upper bound, the dimension of
		\(C(D,G)\cap \mathrm{ev}_D(f)*C(D,H)\)
		is not determined solely by the divisor $\gcd(G,H-(f))$, but also depends on finer global compatibility conditions between multiplication by~$f$ and evaluation on~$D$.
	\end{remark}
	
	\begin{proposition}
		Let $\mathcal{P}$ be a proper subset of $\mathcal{X}(\Fq)$ and set $D = \sum_{P \in \mathcal{P}} P$. 
		Let $G,H,J$ be divisors on $\mathcal{X}$ of degree at most 
		$\deg(D)-1$ such that
		$(\supp(G) \cup \supp(H) \cup \supp(J)) \cap \supp(D)=\varnothing.$
		For each $f \in \L(J)$, let $H_f=H - (f)$.
		If 
		$$\L(\textup{lcm}(G,H_f)-D) \cap (\L(G) + \L(H_f) ) = \{0\} \text{ for all } f \in \L(J),$$
		then $C(D,G) \cap C(D,H_f) = C(D,\textup{gcd}(G,H_f))$.
	\end{proposition}
	
	\begin{proof}
		Let $c \in C(D,G) \cap C(D,H_f)$. Then there exists
		$\mu \in \L(G), \lambda_f \in \L(H_f)$ such that
		$\mu-\lambda_f \in \L(\textup{lcm}(G,H_f)-D) \cap \L(G) +\L(H_f)$. If $\L(\textup{lcm}(G,H_f)-D) \cap \L(G) +\L(H_f)= \{0\}$ this requires that $\mu = \lambda_f$ and hence $\mu \in \L(\textup{gcd}(G,H_f))$.
	\end{proof}

	The following is an immediate consequence of Theorem \ref{thm:generalcurves}.

	\begin{corollary}
		\label{cor:conditions}
		Let $D, G, H, J$ be as in Theorem~\ref{thm:generalcurves}. Let $f \in \mathscr{L}(J)$ and suppose that
		\begin{enumerate}
			\item[(i)] $\ell\, \big(\gcd(G, H - (f))\big) = 0$, and 
			\item[(ii)] $\ell\, \big(\mathrm{lcm}(G, H - (f)) - D\big) = 0$.
		\end{enumerate}
		Then
		
		$C(D,G)\cap \mathrm{ev}_D(f)*C(D,H)=\{0\}.$
		
	\end{corollary}
	
	In this setting, the above conditions can be simplified under suitable hypotheses on the degrees of $J$ and $H$. To this end, we rely on the following well-known results, which can be found in \cite[Theorem~8]{couvreur2017cryptanalysis} or, in a less explicit form, in \cite[Theorem~6]{mumford1970varieties}.
	
	\begin{theorem}\label{th:ShurProductDivisors}
		Let $E, F$ be divisors on $\mathcal{X}$ such that $\deg(E) \geq 2g + 1$ and $\deg(F) \geq 2g$, and let $t$ be a positive integer. Then the following hold.
		\begin{enumerate}
			\item $\mathscr{L}(E) \cdot \mathscr{L}(F) = \mathrm{Span}\{ v w : v \in \mathscr{L}(E),\, w \in \mathscr{L}(F) \} = \mathscr{L}(E + F)$;
			\item $\mathscr{L}(E)^{(t)} = \mathscr{L}(tE)$.
		\end{enumerate}
	\end{theorem}
	
	A direct consequence of this theorem is the following result.
	
	\begin{corollary}\label{cor:ShurProductCodes}
		Let $E, F$ be two divisors on the curve $\mathcal{X}$ both
		having support disjoint from $D$. Suppose that $\deg(E)\ge  2g + 1$
		and $\deg(F)\ge 2g$ and let $t$ be a positive integer. Then,
		
		$C(D,E) * C(D,F) = C(D, E + F)$, and
		
		$C(D, E)(t) = C(D,t E)$. 
		
	\end{corollary}
	
	\begin{proposition}\label{prop:shursimplified}
		Let $\mathcal{P}$ be a proper subset of $\mathcal{X}(\Fq)$ and set $D = \sum_{P \in \mathcal{P}} P$. 
		Let $G,H,J$ be divisors on $\mathcal{X}$ of degree at most 
		$\deg(D)-1$ 
		whose supports are pairwise disjoint from the support of $D$. 
		Suppose that $2g \leq \deg(J), \deg(H)$, and that 
		$\max(\deg(J), \deg(H)) \geq 2g+1$. Then
		$$C(D,G)\cap \mathrm{ev}_D(f)*C(D,H)=\{0\}
		\text{ for all }f \in \mathscr{L}(J),$$
		if and only if
		$C(D,G)\cap C(D,H+J)=\{0\}$.
	\end{proposition}
	
	\begin{proof}
		By Corollary~\ref{cor:ShurProductCodes}, we have
		$C(D,J)*C(D,H)=C(D,H+J).$
		Since
		$C(D,J)=\{\mathrm{ev}_D(f) : f\in \mathscr{L}(J)\},$
		it follows that
		$C(D,H+J)=\{\mathrm{ev}_D(f)*C(D,H) : f\in \mathscr{L}(J)\}.$
		Therefore,
		\[
		C(D,G)\cap C(D,H+J)
		= C(D,G)\cap \{\mathrm{ev}_D(f)*C(D,H) : f\in \mathscr{L}(J)\},
		\]
		which is trivial if $C(D,G)\cap \mathrm{ev}_D(f)*C(D,H)=\{0\}$ for all $f \in \mathscr{L}(J)$. 
		Conversely, if $C(D,G)\cap C(D,H+J)=\{0\}$, then as each subspace 
		$\mathrm{ev}_D(f)*C(D,H)$ is contained in $C(D,H+J)$, we have that
		\[
		C(D,G)\cap \mathrm{ev}_D(f)*C(D,H)=\{0\}
		\quad \text{for each } f \in \mathscr{L}(J).
		\]
		This concludes the proof.
	\end{proof}

	\begin{proposition}\label{prop:maxmin}
		Let $D_1 = \sum_{P \in \mathcal{X}} m_P P$ and $D_2 = \sum_{P \in \mathcal{X}} m_P' P$ be divisors. Then
		\[
		\gcd(D_1, D_2) + \mathrm{lcm}(D_1, D_2) = D_1 + D_2.
		\]
	\end{proposition}
	
	\begin{proof}
		This follows directly from the definitions of $\gcd$ and $\mathrm{lcm}$ of divisors. Indeed, we have
		\[
		\begin{split}
			\gcd(D_1, D_2) + \mathrm{lcm}(D_1, D_2)
			&= \sum_{P} \min(m_P, m_P') P + \sum_{P} \max(m_P, m_P') P \\
			&= \sum_{P} (m_P + m_P') P \\
			&= D_1 + D_2.\qedhere
		\end{split}
		\]
	\end{proof}
	
	The following observations are a consequence of combining Propositions~\ref{prop:shursimplified} and~\ref{prop:maxmin}.
	They give constraints on the degrees of $G,H,J$ for which the conditions of Corollary \ref{cor:conditions} hold.

	\begin{proposition}\label{propn2+g-2}
		Let $D, G, H, J$ be defined as in Theorem~\ref{thm:generalcurves}. Suppose that we have $\deg(J), \deg(H) \geq 2g$ and 
		$\max(\deg(J), \deg(H)) \geq 2g+1$. Let $\deg(D) = n$.
		Suppose that $\ell\, \big(\gcd(G, H - (f))\big) = 0$ and 
		$\ell\, \big(\mathrm{lcm}(G, H - (f)) - D\big) = 0$ for all $f \in \mathscr{L}(J)$. 
		Then
		\[
		\deg(G) + \deg(H) + \deg(J) \leq n + 2g - 2.
		\]   
	\end{proposition}
	
	\begin{proof}
		By hypothesis, for any $f \in \mathscr{L}(J)$, the upper bound provided by Theorem~\ref{thm:generalcurves} is zero. Therefore, the dimension of the intersection is zero, meaning:
		\[
		C(D,G)\cap \mathrm{ev}_D(f)*C(D,H)=\{0\} \quad \text{for all } f \in \mathscr{L}(J).
		\]
		
		By Proposition \ref{prop:shursimplified} and Corollary \ref{cor:conditions}, this implies that the intersection with the full product space is also trivial, i.e., 
		$C(D,G)\cap C(D,H+J)=\{0\}.$
		
		As $C(D,G)$ and $C(D,H+J)$ are subspaces of $\mathbb{F}_q^n$ that intersect trivially, the dimension of their sum is simply the sum of their dimensions, which cannot exceed $n$:
		\[
		\dim(C(D,G)) + \dim(C(D,H+J)) \le n.
		\]
		
		We have $\dim(C(D,G)) \ge \deg(G) + 1 - g$ and $\dim(C(D,H+J)) \ge \deg(H+J) + 1 - g$, by the Riemann-Roch theorem. 
		Substituting these lower bounds into the previous inequality yields:
		\begin{eqnarray*}
			n &\geq& (\deg(G) + 1 - g) + (\deg(H+J) + 1 - g) \\
			&=& \deg(G) + \deg(H) + \deg(J) + 2 - 2g.
		\end{eqnarray*}
		
		Rearranging the terms, we obtain the desired inequality:
		\[
		\deg(G) + \deg(H) + \deg(J) \leq n + 2g - 2,
		\]
		which concludes the proof.
	\end{proof}
	
	\begin{remark}
		Let $D, G, H, J$ be defined as in Theorem~\ref{thm:generalcurves}, and assume that $\deg(G) \geq 2g-1$, that $\deg(J), \deg(H) \geq 2g$, $C(D,G)\cap C(D,H+J)=\{0\}$, and that
		$\max(\deg(J), \deg(H)) \geq 2g+1$. Let $\deg(D) = n$.
		Assume now that
		$
		\deg(G) + \deg(H) + \deg(J) = n\text{ and }\ell\big(\mathrm{lcm}(G, H+J) - D\big) = 0.
		$
		Again, using Proposition \ref{prop:maxmin} and the Riemann-Roch theorem, we deduce that $\ell\big(\gcd(G, H+J)\big) = 0.$
		By the Riemann-Roch theorem, this forces 
		$\deg\big(\gcd(G, H+J)\big) \leq g - 1$. 
		On the other hand, using Proposition~\ref{prop:maxmin} and the assumption $\ell\big(\mathrm{lcm}(G, H+J) - D\big) = 0$, we obtain
		\[
		n - (g - 1)
		\;\leq\;
		\deg\big(\mathrm{lcm}(G, H+J)\big)
		\;\leq\;
		n + g - 1.
		\]
		
		In this case, the overlap between $G$ and $H+J$ is necessarily small (of degree at most $g - 1$), which constrains how the degrees of $G$, $H$, and $J$ can be distributed.
	\end{remark}

	\begin{notation}
		For the remainder of this section, let $D, G, H, J$ be divisors on an elliptic curve $\mathcal{E}$, with $\supp(D)$ disjoint from $\supp(G)$, $\supp(H)$, and $\supp(J)$.  Let $V$ be a parity-check matrix of the code $C(D,G)$, let $W$ be a generator matrix of the code $C(D,H)$, and let $C = C(D,J)$. Then, we have 
		\[
		C_V^\perp = C(D,G), \qquad C_W = C(D,H), \qquad C = C(D,J).
		\]
	\end{notation}
	
	We now provide an example of codes arising from elliptic (that is, $g=1$) curves that have trivial intersection but do not satisfy some of the conditions in Corollary~\ref{cor:conditions}. Specifically, these codes satisfy $\deg(\gcd(G, H+J)) = 0$ and $C(D,G)\cap \, \mathrm{ev}_D(f)*C(D,H)=\{0\}$, but $\ell\big(\mathrm{lcm}(G, H+J)\big) \ne 0$.
	
	\subsection{Elliptic curves}
	If the curve $\mathcal{X}=\mathcal{E}$ is elliptic, the conditions (i) and (ii) of 
	Corollary \ref{cor:conditions}
	can be expressed as follows, since in this case a canonical divisor is known to be linearly equivalent to $W=0$. Let $\deg(G)=n$.
	\begin{enumerate}
		\item[$(a)$] For every $f \in \L(J)$, $\deg(\gcd(G, H-(f))\big) \le 0$.
		
		\item[$(b)$] For every $f\in\L(J)$, $ \deg(\mathrm{lcm}(G,H-(f))) \le n$.
	\end{enumerate}
	
	By Corollary~\ref{cor:ShurProductCodes}, when 
	$\deg(J), \deg(H) \geq 2$ and $\max\{\deg(J), \deg(H)\} \geq 3$, 
	conditions $(a)$ and~$(b)$ can be rewritten as follows.
	
	\begin{itemize}
		\item[$(a')$] 
		$\deg(\gcd(G,H+J)) \le 0$.
		
		\item[$(b')$] 
		$\deg(\mathrm{lcm}(G,H+J)) \le n$.
	\end{itemize}

	\begin{example}
		Let $\mathcal{E}$ be the elliptic curve $\mathcal{E}:y^2+4x^3+4x+4=0$. \\Let $G=2\,(0:4:1)-(0:1:0)$, $H=(0:4:1)$, $J=4\,(0:1:0)$ and 
		\begin{equation*}
			D=(0 : 1 : 1)+(2 : 4 : 1)+(2 : 1 : 1)+(4 : 3 : 1)+(4 : 2 : 1)+(3 : 4 : 1)+(3 : 1 : 1).
		\end{equation*}
		Note that $(\supp(G)\cup\supp(H)\cup \supp(J))= \varnothing$. Observe that, for any $f\in J$, we have $\ell(\textup{gcd}(G,H-(f)))=0$, $\ell(\textup{lcm}(G,H-(f))-D)=0$ and
		\begin{equation*}
			C(D,G)\cap \mathrm{ev}_D(f)*C(D,H)=\{0\},
		\end{equation*}
		in line with Theorem \ref{thm:generalcurves}.
	\end{example}

	We present an explicit construction of divisors $G$, $H$, and $J$ with small support on an elliptic curve such that conditions $(a')$ and $(b')$ are satisfied.
	
	\begin{example}
		Let $\mathcal{E}$ be the elliptic curve defined by the polynomial $y^2=x^3+x+1\in\mathbb{F}_{11}[x,y]$. Setting $P_1=(0 : 1 : 1)$ and $P_2=(0 : 1 : 0)$, consider the divisors
		\[
		\begin{split}
			G&=-2P_1+ 5P_2,\\
			H&=3P_1,\\
			J&=3P_1+P_2,\\
			D&=\sum_{P\in\mathcal{E}(\mathbb{F}_{11})\setminus\{P_1,P_2\}}P.
		\end{split}
		\] 
		It can be checked that $\deg(D)=n=12$, that 
		$C_V=C(D,G)^\perp$ is an $\mathbb{F}_{11}$-$[12,9,3]$ code, that $C_W=C(D,H)$ is an $\mathbb{F}_{11}$-$[12,3,9]$ code, and that $C=C(D,J)$ is an $\mathbb{F}_{11}$-$[12,4,8]$ code.
		
		The trivial intersection property of these codes is entirely governed by the degrees of $G$ and $H+J$. We have that
		\[
		\begin{split}
			\deg(\gcd(G,H+J)) &= \deg(-2P_1+P_2) = -1 < 0, \text{and} \\
			\deg(\mathrm{lcm}(G,H+J)) &= \deg(6P_1+5P_2) = 11 = n-1.
		\end{split}
		\]
		These conditions guarantee that $\dim(C_V^\perp \cap (c*C_W)) = 0$ for all $c \in C$. It follows that $\mathcal{C}=\Phi_{V,W}(C)$ yields an $\F_{11}$-$[9\times 3, 4,3]$ rank-metric code.
	\end{example}

	In the next proposition, a triple of elliptic codes satisfying the above properties can be constructed by suitably adapting the construction in \cite[Section~6]{bhowmick2025}. We omit the proof, as it follows directly from Theorem~\ref{thm:generalcurves}. Indeed, for the divisors given in Proposition \ref{th:abconstr}, one simply observes that
	$
	\ell(\operatorname{lcm}(G,H+J)-D)=0,
	$
	which implies
	\[
	\dim\bigl(C(D,G)\cap C(D,H+J)\bigr)
	= \dim\bigl(C(D,\gcd(G,H+J))\bigr).
	\]

	\begin{proposition}\label{th:abconstr}
		Let $\mathcal{E}$ be a nonsingular elliptic curve defined over $\mathbb{F}_q$, and let $P_\infty$ denote the point at infinity. Let $P_0 \in \mathcal{E}(\mathbb{F}_q)$ be a rational point with $P_0 \neq P_\infty$. Let $a,b$ be positive integers satisfying $a+\ell < b-\ell < n$. Define the divisor
		\[
		D = \sum_{P \in \mathcal{E}(\mathbb{F}_q) \setminus \{P_0,P_\infty\}} P,
		\]
		with $n = \deg(D)$ and let
		\[
		G = (b-\ell)P_0 - aP_\infty, \qquad
		H = (n-b)P_\infty, \qquad
		J = (a+\ell)P_0.
		\]
		
		Then $\dim\bigl(C(D,G) \cap C(D,H+J)\bigr) = \ell$.
	\end{proposition}
	
	As these constructions illustrate, the key point is the ability to control the behaviour of the divisors $G$ and $H+J$. Once their supports and coefficients are chosen to satisfy the numerical conditions imposed by Proposition~\ref{prop:shursimplified}, the dimensions of the intersections of the associated AG codes follow immediately. 
	
	\begin{corollary}\label{cor:abconstr}
		Let $a,b$ be positive integers satisfying $a+\ell < b-\ell < n$, for a non-negative integer $\ell$, and let $D,G,H,J$ be as defined in Theorem~\ref{th:abconstr}. Let $C = C(D,J)$, let $W$ be a generator matrix for $C(D,H)$, and let $V$ be a parity-check matrix for $C(D,G)$. Then $\mathcal{C}= \phi_{V,W}(C)$ is an $\mathbb{F}_q$-$[(n-b+\ell+a) \times (n-b),\, a+\ell,\, d \geq n-b-\ell]$ rank-metric code. 
		In particular, 
		$n \geq \trk(\mathcal{C}) \geq n-(b-a+1)$.
	\end{corollary}
	
	\begin{proof}
		First note that $C_V$ is an $\Fq$-$[n,n-b+\ell+a,b-\ell-a]$ code, being the dual of $C(D,G)$. Moreover, $C_W = C(D,H)$ is an $\Fq$-$[n,n-b,b]$ code, and $C = C(D,J)$ is an $\Fq$-$[n,a+\ell,n-a-\ell]$ code. We have $C * C_W = C(D,H+J)$, and thus, by Theorem~\ref{th:abconstr}, $
		\dim(C_V^\perp \cap C * C_W) = \ell$. 
		Since $\dim(C_V^\perp) = \dim(C(D,G))=b-a-\ell > \ell \geq \dim(C_V^\perp \cap c*C_W)$ for any $c \in C$, we have that $C_V^\perp \not\subseteq c* C_W$ for any $c \in C$. Hence $\phi_{V,W}$ is an isomorphism and so $\dim(\C)=\dim(C)=a+\ell$.
		For any $c \in C$, we have
		\begin{multline*}
			\rk\bigl(V_cW^{\textup{t}}\bigr) 
			= \rk\bigl(W_c\bigr) - \dim\bigl(C_V^\perp \cap c * C_W\bigr)
			\geq \rk\bigl(W_c\bigr) - \dim\bigl(C_V^\perp \cap C * C_W\bigr)
			= \rk\bigl(W_c\bigr) - \ell.
		\end{multline*}

		Since $a+\ell < b-\ell$, we have 
		$\dH(C) = n-a-\ell \geq n-b+1 = n - \dH(C_W) + 1$, and hence $\rk(W) = n-b = \rk\bigl(W_c\bigr)$ for all  $c \in C$, from which we deduce that $\drk(\C)=n-b-\ell$. Then, by Kruskal's bound we have
		\[
		n \geq \trk(\mathcal{C}) \geq n-(b-a+1).
		\]
		
		If $\ell = 0$, then for any $c \in C$, we have $c * C_W \nsubseteq C_V^\perp$ unless $c * C_W = \{0\}$, which holds if and only if $c = 0$. It follows that $\ker \phi_{V,W} \cap C = \{0\}$, and hence $k = 0$. Therefore, $\mathcal{C}$ is an~$\Fq$-$[(n-b+a) \times (n-b), a, n-b]$ code.
	\end{proof}
	
	For the case $\ell = 0$, assuming that $\gcd(G,H+J)$ and 
	$\operatorname{lcm}(G,H+J)$ are non-special, Corollary~\ref{cor:abconstr} provides a construction of a rank-metric code $\mathcal{C}$ whose tensor rank defect is bounded from above, that is, $\trkdef(\mathcal{C}) \leq b-a+1$. In particular, for $a = b-1$ such codes have tensor rank defect at most $2$. 
	\begin{example}
		Let $n=12, a=3, b=7$, and $\ell=0$. Then $a < b < n$.
		Let $\mathcal{E}$ be the elliptic curve defined by the polynomial $y^2+x^3+x+1\in\mathbb{F}_{11}[x,y]$, and consider the divisors
		\[
		\begin{split}
			G=&-3(0 : 1 : 0)+ 7(0 : 1 : 1),\\
			H=&5(0 : 1 : 0),\\
			J=&3(0 : 1 : 1),\\
			D=&(0 : 10 : 1)+(2 : 0 : 1)+(4 : 6 : 1)+(4 : 5 : 1)+(8 : 9 : 1)+(8 : 2 : 1)+(3 : 8 : 1)+(3 : 3 : 1)+\\&(6 : 6 : 1)+(6 : 5 : 1)+(1 : 6 : 1)+(1 : 5 : 1),
		\end{split}
		\]
		which satisfy the conditions of Theorem \ref{th:abconstr}.
		It can be checked that both $\gcd(G,H+J)$ and $\mathrm{lcm}(G,H+J)$ are non-special divisors, so we obtain a triple of codes $C_V^\perp=C(D,G)$, $C_W=C(D,H)$ and $C=C(D,J)$, which have parameters $[12,4,8]$, $[12,5,7]$, and $[12,3,9]$, respectively. Furthermore, we have 
		\[
		C_V^\perp \cap (C* C_W) = C(D,G) \cap C(D,H+J)=\{0\}.
		\]
		By Corollary \ref{cor:abconstr}, 
		the code $\C=\phi_{V,W}(C)=\{V_c W^{\textup{t}}: c \in C\}$ is an 
		$\F_{11}$-$[8\times 5, 3, 4]$ rank-metric code whose tensor rank satisfies
		$12 \geq \trk(\C) \geq 4+3-1 = 6.$

	\end{example}

	While the previous results regarding the tensor rank defect provide a nice general framework, it is possible in certain cases to obtain even more precise outcomes. The following example illustrates a construction where we achieve an optimal tensor rank by leveraging the geometry of an elliptic curve.
	
	\begin{example}
		\label{ex:zerointer}
		Let $\mathcal{E}$ be the elliptic curve defined over $\mathbb{F}_7$ by the equation $6x^3 + y^2 + 4 = 0$, and consider the following divisors.
		\begin{align*}
			G=&4(0 : 1 : 0)-7(3 : 4 : 1) + 5(3 : 3 : 1),\\
			H=&5(0 : 1 : 0) + 3(3 : 4 : 1),\\
			J=&2(0 : 1 : 0),\\
			D=&(2 : 5 : 1) +(2 : 2 : 1) +(6 : 4 : 1) +(6 : 3 : 1) +
			(4 : 5 : 1) +(4 : 2 : 1) +(5 : 4 : 1) +(5 : 3 : 1) \\&+(1 : 5 : 1) +(1 : 2 : 1).
		\end{align*}
		It is immediate to verify that $(\supp(G)\cup \supp(H)\cup \supp(J)) \cap \supp(D) = \varnothing$. We have that $C_V^\perp = C(D,G)$ is an $\mathbb{F}_7$-$[10,2,8]$ code, $C = C(D,J)$ is an $\mathbb{F}_7$-$[10,2,8]$ code, and $C_W=C(D,H)$ is an $\mathbb{F}_7$-$[10,8,2]$ code, which have the following generator matrices, respectively,
		{\footnotesize
			\begin{align*}
				\left(\begin{array}{*{2}c}
					1&0\\
					0&1\\
					1&0\\
					6&3\\
					6&1\\
					3&3\\
					6&4\\
					1&4\\
					2&2\\
					1&3\\
				\end{array}\right)^{t}\quad, \quad
				\left(\begin{array}{*{2}c}
					1&0\\
					1&0\\
					0&1\\
					0&1\\
					4&4\\
					4&4\\
					2&6\\
					2&6\\
					3&5\\
					3&5\\
				\end{array}\right)^{t}
				\quad,\quad
				\left(\begin{array}{*{10}c}
					1&0&0&0&0&0&0&0&2&4\\
					0&1&0&0&0&0&0&0&5&0\\
					0&0&1&0&0&0&0&0&1&4\\
					0&0&0&1&0&0&0&0&5&4\\
					0&0&0&0&1&0&0&0&5&6\\
					0&0&0&0&0&1&0&0&5&1\\
					0&0&0&0&0&0&1&0&4&4\\
					0&0&0&0&0&0&0&1&2&6
				\end{array}\right),
			\end{align*}
		}
		while $C_V$ is the $\mathbb{F}_7$-$[10,8,2]$ code generated by
		\begin{equation*}
			\left(
			\begin{array}{*{10}c}
				1&0&0&0&0&0&0&0&1&4\\
				0&1&0&0&0&0&0&0&2&3\\
				0&0&1&0&0&0&0&0&1&4\\
				0&0&0&1&0&0&0&0&5&5\\
				0&0&0&0&1&0&0&0&1&6\\
				0&0&0&0&0&1&0&0&2&0\\
				0&0&0&0&0&0&1&0&0&1\\
				0&0&0&0&0&0&0&1&2&2
			\end{array}
			\right).
		\end{equation*}
		We have that $\ell\big(\gcd(G, H+J)\big) = 0$
		and $\ell\big(\mathrm{lcm}(G, H+J) - D\big) = 5$. On the other hand, it can be verified that
		\[
		C(D,G)\cap \mathrm{ev}_D(f)*C(D,H)=\{0\}
		\qquad \text{for all } f \in \mathscr{L}(J).
		\]
		Consider the matrix code $\C = \phi_{V,W}(C)=\{ V \,\mathrm{diag}(c)\, W^{\textup{t}} : c \in C \} \subseteq \Fq^{8 \times 8}$. 
		Since $c*C_W \cap C_V^\perp = \{0\}$ for all $c\in C$, we have that $\phi_{V,W}$ is an isomorphism and furthermore that $\rk(\phi(c)) = \rk(W _c)$ for each $c \in C$.

		It is easy to check that $\rk(W_c)=\rk(W)=8$ for each $c \in C$ and so we have that
		$\C$ is an $\mathbb{F}_7$-$[8 \times 8, 2, 8]$ rank-metric code. It generated by the following matrices.
		\begin{equation*}
			\begin{pmatrix}
				1&0&0&0&0&0&0&0\\
				0&1&0&0&0&0&0&0\\
				0&0&5&0&0&0&0&0\\
				0&0&0&5&0&0&0&0\\
				0&0&0&0&3&0&0&0\\
				0&0&0&0&0&3&0&0\\
				0&0&0&0&0&0&4&0\\
				0&0&0&0&0&0&0&4
			\end{pmatrix},\quad 
			\begin{pmatrix}
				0&1&2&0&3&6&4&1\\
				6&4&0&3&0&4&4&3\\
				5&1&0&0&3&6&4&1\\
				6&5&5&0&4&6&1&1\\
				1&1&5&3&4&5&0&2\\
				5&2&6&2&2&2&3&5\\
				5&0&5&5&4&3&4&4\\
				1&2&2&5&3&1&6&5
			\end{pmatrix}.   
		\end{equation*}
		By construction, we have $\trk(\C) \leq 10$, which is the length of the codes $C_V$, $C_W$, and $C$. 
		Since~$\C$ is an $\F_7$-$[8 \times 8, 2, 8]$ code, by Kruskal's bound we have $\trk(\C) \geq 2 + 8 - 1 = 9$. However, in order to attain this bound, $\C$ must be MTR, which would require the existence of an MDS code with parameters $\F_7$-$[9,2,8]$. By \cite{ball2012sets}, this is impossible. It follows that $\trk(\C)=10$ and $\Delta^\trk(\C)=1$, which is optimal, that is, $\trk(\C)=N_7(2,8)$.
	\end{example}
	
	Using the same approach as above, one can likewise construct a rank-metric code of dimension $4$ arising from the same framework.

	\begin{example}
		\label{ex:GF13_nonzero_lcm}
		Let $\mathcal{E}$ be the elliptic curve over $\mathbb{F}_{13}$ given in  affine coordinates by $\mathcal{E}:y^2=x^3+4$.
		Denote the $\mathbb{F}_{13}$-rational places corresponding to the points
		\[
		P_\infty=(0:1:0),\qquad P_1=(0:2:1),\qquad P_{2}=(0:11:1).
		\]
		Consider the following divisors on $\mathcal{E}$:
		\begin{align*}
			G&= -15P_\infty + 4P_1 + 14P_{2},\\
			H&= 6P_\infty + 13P_1 - 4P_{2},\\
			J&= 3P_\infty.
		\end{align*}
		Let $D$ be the sum of all $\mathbb{F}_{13}$-rational places of $\mathcal{E}$
		excluding $P_1,P_2$.
		Since $\#\mathcal{E}(\mathbb{F}_{13})=21$ and $\supp(G)\cup\supp(H)\cup\supp(J)=\{P_\infty,P_1,P_{2}\}$,
		we have $\deg(D)=18$ and
		
		$(\supp(G)\cup\supp(H)\cup\supp(J))\cap\supp(D)=\varnothing$.
		
		Let
		$C_V^\perp=C(D,G), C=C(D,J)$, and $C_W=C(D,H)$;
		then $C_V^\perp$ is an $\F_{13}$-$[18,3,15]$ code, $C$ is an $\F_{13}$-$[18,3,15]$, and $C_W$ is an $\F_{13}$-$[18,15,3]$ code.
		Let $\C=\phi_{V,W}(C)$, which is an $\F_{13}$-$[15 \times 15, 3,15]$ code.

		As before,
		Kruskal's bound yields that $\deg(D)=18\ge\trk(\mathcal{C})\ge 3+15-1=17$.
		If $\trk(\mathcal{C})=17$, then $\mathcal{C}$ is MTR, forcing the existence of an
		$\mathbb{F}_{13}$-$[17,3,15]$ MDS code, which is impossible since nontrivial MDS codes over
		$\mathbb{F}_{13}$ have length at most $14$.
		Hence $\trk(\mathcal{C})=18$ and $\Delta^{\trk}(\mathcal{C})=1$ is optimal.
	\end{example}
	
	The following theorem generalizes this construction to a broader family of rank-metric codes arising from the same setup.
	
	\begin{theorem}
		\label{thm:general_construction_AMTR}
		
		Let $\mathcal{E}$ be an elliptic curve defined over $\mathbb{F}_q$ with a set of rational points
		$\mathcal{P} \subseteq \mathcal{E}(\mathbb{F}_q)$.
		Let $D$ be a divisor defined as the sum of $n$ distinct points in $\mathcal{P}$.
		Let $G,H,J$ be divisors on $\mathcal{E}$ such that
		$
		(\supp(G)\cup\supp(H)\cup\supp(J))\cap\supp(D)=\varnothing.
		$
		Let $C_V^\perp = C(D,G)$, let $C = C(D,J)$, and let $C_W = C(D,H)$.
		Assume the parameters satisfy $\dH(C)=\dim(C_W)=\dim(C_V)=m$ and $\dim(C)=k$, with $k=n-m$.
		For every $f\in \mathscr{L}(J)\setminus\{0\}$, 
		let $H_f = H-(f)$ and define $D_f = \gcd\big(D,(f)\big)$ and $E_f = D-D_f$.
		Assume that for every $f\in \mathscr{L}(J)\setminus\{0\}$ the following hold:
		\begin{itemize}
			\item $\ell(\gcd(G,H_f))=0$,
			\item $\mathscr{L}(\mathrm{lcm}(G,H_f)-D)\cap(\mathscr{L}(G)+\mathscr{L}(H_f))=\{0\}$, and
			\item $\ell(H-E_f)=\ell \big(H-(D-\gcd(D,(f)))\big)=0.$
		\end{itemize}

		Suppose the MDS conjecture holds over $\mathbb{F}_q$ and $n>q+1$.
		Then the matrix code $\mathcal{C}=\phi_{V,W}(C)$ 
		is an $\Fq$-$[m \times m, n-m,m]$ code such that $\trk(\mathcal{C})=n$ and $\Delta^{\trk}(\mathcal{C})=1$. In particular, $\C$ is an optimal AMTR matrix code.
		
	\end{theorem}
	
	\begin{proof}
		Note first that by construction $C_V^\perp$ and $C_W$ are $[n,m,n-m]$ codes, while $C$ is an $[n,n-m,m]$ code,
		and that they are NMDS codes, as they are elliptic AMDS codes.
		
		By the assumed conditions on the divisors, we have that
		\[
		C(D,G)\cap \mathrm{ev}_D(f)*C(D,H)=\mathscr{L}(\gcd(G,H_f))=\{0\}
		\]
		for all $f\in \mathscr{L}(J)\setminus\{0\}$.
		This implies that for every nonzero codeword $c\in C$ we have $C_{W_c}\cap C_V^\perp=\{0\}$,
		and hence $\rk(\phi_{V,W}(c))=\rk(W_c)$ for each $c\in C$.
		
		It remains to show that $\rk(W_c)=m$ for every $c\in C\setminus\{0\}$.
		Since $W$ generates $C_W$, we have $\rk(W) \geq \rk(W_c)=\dim(C_W*c)$.
		
		If $\rk(W) > \rk(W_c)$ for some nonzero $c \in C$, then there exists $h \in \L(H)$ and 
		$c'=\mathrm{ev}_D(h)\in C_W$ such that $c*c'=0$, which holds if and only if
		$f(P)h(P)=0 $ for all $P\in\supp(D).$
		Since $\supp(J)\cap\supp(D)=\varnothing$, the function $f$ has no poles on $\supp(D)$, and therefore
		$D_f=\gcd(D,(f))$
		is exactly the sum of those $P\in\supp(D)$ for which $f(P)=0$. In particular, $f(P)\neq 0$ for every $P\in\supp(E_f)$,
		where $E_f=D-D_f$.
		Hence, from $f(P)h(P)=0$ for all $P\in\supp(D)$ we obtain that
		\[
		h(P)=0 \qquad \text{for all }P\in\supp(E_f),
		\]
		so $h\in \mathscr{L}(H-E_f)$. By the additional assumption $\ell(H-E_f)=0$ we conclude that $h=0$,
		which contradicts $c'\neq 0$. Therefore there is no nonzero $c'\in C_W$ with $c*c'=0$, and consequently
		\[
		\dim(C_W*c)=m\qquad\text{and}\qquad \rk(W_c)=m
		\]
		for all $c\in C\setminus\{0\}$.
		Therefore, $\mathcal{C}$ is an $\mathbb{F}_q$-$[m\times m,n-m,m]$ rank-metric code.
		
		By construction, the tensor rank of $\mathcal{C}$ is bounded from above by the length of the underlying codes,
		so $\trk(\mathcal{C})\le \deg(D)=n$.
		Applying Kruskal's bound to $\mathcal{C}$, we obtain
		$\trk(\mathcal{C})\ge n-m+m-1=n-1.$
		If $\trk(\mathcal{C})=n-1$, then $\mathcal{C}$ is an MTR code.
		However, the existence of an $\mathbb{F}_q$-$[m\times m,n-m,m]$ MTR code implies the existence of an MDS code over $\mathbb{F}_q$
		with parameters $[n-1,n-m,m]$. Since we have assumed that the MDS conjecture holds and
		that $n>q+1$, no such MDS code exists, which implies that $\mathcal{C}$ is not MTR.
		It follows that $\trk(\mathcal{C})=n$, and the tensor rank defect is $\Delta^{\trk}(\mathcal{C})=1$.
	\end{proof}

	\begin{remark}
		The hypothesis \(\ell(H-E_f)=0\) for all \(f\in\mathscr{L}(J)\setminus\{0\}\) can be replaced by a simpler sufficient condition.
		Indeed, since \(\supp(J)\cap\supp(D)=\varnothing\), for every such \(f\) the effective divisor \(E_f=D-\gcd(D,(f))\) is supported on the set of points of \(D\) where \(f\) does \emph{not} vanish. Hence one may impose the easier-to-check requirement
		\[
		\ell \big(H-(D-\supp(J))\big)=0,
		\]
		which guarantees \(\ell(H-E_f)=0\) uniformly for all \(f\in\mathscr{L}(J)\setminus\{0\}\).
		
		Notice also that if this condition does not hold, the proof above remains valid up to the final conclusion on \(\rk(W_c)\): one may then have \(\rk(W_c)=m-1\) for some \(c\in C\setminus\{0\}\), but never a greater decrease from $\rk(W)$. Consequently the constructed matrix code \(\mathcal{C}=\phi_{V,W}(C)\) still has tensor rank at least \(n-2\), and therefore its tensor rank defect satisfies $\Delta^{\trk}(\mathcal{C})\le 2$.
	\end{remark}

	\begin{remark}
		As the number of rational points on $\mathcal{E}$ increases, the size of the ambient space $\mathbb{F}_q^n$ grows relative to the dimensions of the involved Riemann-Roch spaces. This makes the condition $C(D,G) \cap \mathrm{ev}_D(f) * C(D,H) = \{0\}$ significantly easier to satisfy, allowing for the construction of such codes for higher dimensions $k$.
	\end{remark}

	\section*{Acknowledgments}
	The authors would like to thank Beatriz Barbero Lucas for helpful discussions, and the First Dublin Discrete Mathematics Workshop, where some of the topics of this paper were discussed.
	
	The first author was supported by \textit{Blue Mathematics at AAU}, funded by Orient's Fond.

	\newpage
	
	\bibliographystyle{abbrv}
	\bibliography{refs}

\end{document}